\documentclass[12pt,journal,onecolumn]{IEEEtran}
\usepackage{amssymb}
\usepackage{amsmath}
\usepackage{longtable}
\newtheorem{theorem}{Theorem}[section]
\newtheorem{lemma}[theorem]{Lemma}

\newtheorem{example}{Example}

\allowdisplaybreaks

\usepackage[para]{threeparttable}

\newcommand{\tr}{{\mathrm{Tr}}}

\newcommand{\gf}{{\mathrm{GF}}}
\newcommand{\PG}{{\mathrm{PG}}}

\newcommand{\wt}{{\mathtt{wt}}}

\newcommand{\F}{{\mathbb{F}}}

\newcommand{\bC}{{\mathbb{C}}}

\newcommand{\C}{{\mathcal{C}}}
\newcommand{\T}{{\mathcal{T}}}
\newcommand{\V}{{\mathcal{V}}}

\newcommand{\cS}{{\mathcal{S}}}

\newcommand{\bs}{{\mathbf{s}}}

\newcommand{\bc}{{\mathbf{c}}}

\newcommand{\bx}{{\mathbf{x}}}

\newcommand{\cV}{{\mathcal{V}}} 
\newcommand{\cO}{{\mathcal{O}}} 
\newcommand{\cT}{{\mathcal{T}}}

\newenvironment{proof}[1][Proof]{\begin{trivlist}
\item[\hskip \labelsep {\bfseries #1}]}{\end{trivlist}}

\makeatletter

\newcommand{\Rmnum}[1]{\expandafter\@slowromancap\romannumeral #1@}
\makeatother
\newcommand\myatop[2]{\genfrac{}{}{0pt}{}{#1}{#2}}

\ifCLASSINFOpdf

\else

\fi

\begin{document}
%
\title{The Subfield Codes of Ovoid Codes
}

\author{ Cunsheng Ding\thanks{C. Ding is with the Department of Computer Science
                           and Engineering, The Hong Kong University of Science and Technology,                                                  Clear Water Bay,
                           Kowloon, Hong Kong, China (email: cding@ust.hk)},
 Ziling Heng\thanks{Z. Heng is with the Department of Computer Science
                           and Engineering, The Hong Kong University of Science and Technology,                                                  Clear Water Bay,
                           Kowloon, Hong Kong, China (email: zilingheng@163.com)}
}

\maketitle

\begin{abstract}
Ovoids in $\PG(3, \gf(q))$ have been an interesting topic in coding theory, combinatorics, 
and finite geometry for a long time. So far only two families of ovoids are known. The first 
is the elliptic quadratics and the second is the Tits ovoids. It is known that an ovoid  
in $\PG(3, \gf(q))$ corresponds to a $[q^2+1, 4, q^2-q]$ code over $\gf(q)$, which is called 
an ovoid code. The objectives of this paper is to study the subfield codes of the two families of 
ovoid codes. The dimensions, minimum weights, and the weight distributions of the subfield 
codes of the elliptic quadric codes and Tits ovoid codes are settled. The parameters of 
the duals of these subfield codes are also studied. Some of the codes presented in this 
paper are optimal, and some are distance-optimal. The parameters of the subfield codes 
are new.

\end{abstract}


%
\IEEEpeerreviewmaketitle

\section{Introduction}

 Let $q$ be a prime power. Let $n, k, d$ be positive integers.
 An $[n,\, k,\, d]$ \emph{code} $\C$ over $\gf(q)$ is a $k$-dimensional subspace of $\gf(q)^n$ with minimum
 (Hamming) distance $d$.
 Let $A_i$ denote the number of codewords with Hamming weight $i$ in a code
 $\C$ of length $n$. The {\em weight enumerator} of $\C$ is defined by
 $1+A_1z+A_2z^2+ \cdots + A_nz^n.$
 The sequence $(1, A_1, A_2, \cdots, A_n)$ is called the \emph{weight distribution} of the code $\C$.
 A code $\C$ is said to be a $t$-weight code  if the number of nonzero
 $A_i$ in the sequence $(A_1, A_2, \cdots, A_n)$ is equal to $t$. 
 An $[n,\, k,\, d]$ code over $\gf(q)$ is called \emph{distance-optimal} if there is no 
 $[n,\, k,\, d+1]$ code over $\gf(q)$, and \emph{dimension-optimal} if there is no 
 $[n,\, k+1,\, d]$ code over $\gf(q)$. A code is said to be optimal if it is both 
 distance-optimal and dimension-optimal.

A cap in the projective space $\PG(3, \gf(q))$ is a set of points in $\PG(3, \gf(q))$ such that 
no three of them are collinear. Let $q>2$. For any cap $\cV$ in $\PG(3, \gf(q))$, 
we have $|\cV| \leq q^2+1$ (see \cite{Bose47}, \cite{Seiden50} and \cite{Qvist52} 
for details).   
In the projective space $\PG(3, \gf(q))$ with $q>2$, an \emph{ovoid}\index{ovoid} $\cV$ 
is a set of $q^2+1$ points such that no three of them are collinear (i.e., on the same line). In other words, an ovoid is a $(q^2+1)$-cap (a cap with $q^2+1$ points) in 
$\PG(3, \gf(q))$, and thus a maximum cap. 

A \emph{classical ovoid} $\cV$ can be defined as the following set of points:  
\begin{eqnarray}\label{eqn-ellipticquadric}
\cV=\{(0,0,1, 0)\} \cup \{(x,\, y,\, x^2+xy +ay^2,\, 1): x,\, y \in \gf(q)\}, 
\end{eqnarray} 
where $a \in \gf(q)$ is such that the polynomial $x^2+x+a$ has no root in $\gf(q)$. 
Such ovoid is called an \emph{elliptic quadric}\index{quadric}, as the points 
come from a non-degenerate elliptic quadratic form.  

For $q=2^{2e+1}$ with $e \geq 1$, there is an ovoid which is not an elliptic quadric, 
and is called the \emph{Tits oviod}\index{Tits ovoid} \cite{Tits60}. It is defined by 
\begin{eqnarray}
\cT=\{(0,0,1,0)\}\cup \{(x,\,y,\, x^{\sigma} + xy +y^{\sigma+2},\,1): x, \, y \in \gf(q)\}, 
\end{eqnarray}   
where $\sigma=2^{e+1}$. 

For odd $q$, any ovoid is an elliptic quadric (see \cite{Barlotti} and \cite{Panella}).
For even $q$, Tits ovoids are the only known ones which are not elliptic quadratics. 
In the case that $q$ is even, 
the elliptic quadrics and the Tits ovoid are not equivalent \cite{Willems}. 
For further information about ovoids, the reader is referred to \cite{OKeefe}. 

Let $\cS$ be a subset of $\PG(3, \gf(q))$ with $q^2+1$ elements, where $q>2$. Denote by 
$$ 
\cS=\{\bs_1, \bs_2, \cdots, \bs_{q^{2}+1}\} 
$$ 
where each $\bs_i$ is a column vector in $\gf(q)^4$. Let $\C_{\cS}$ be the linear code over $\gf(q)$ 
with generator matrix 
\begin{eqnarray}
G_{\cS}=\left[\bs_1 \bs_2 \cdots \bs_{q^{2}+1}\right].  
\end{eqnarray}  

The following result is known (see \cite[p. 192]{Bier} or \cite{Ding18arxiv}). 

\begin{theorem}\label{th-dualdistanceofC}
The set $\cS$ is an ovoid if and only if $\C_{\cS}$ has parameters $[q^2+1, 4, q^2-q]$. 
\end{theorem} 

Due to Theorem \ref{th-dualdistanceofC}, any $[q^2+1, 4, q^2-q]$ code over $\gf(q)$ is 
called an \emph{ovoid code.} Ovoid codes are optimal, as they meet the Griesmer bound. 
It is also known that any $[q^2+1, 4, q^2-q]$ code over $\gf(q)$ must have the following 
weight enumerator \cite[p. 192]{Bier}: 
\begin{eqnarray}\label{eqn-wtdistovidcode}
1+(q^2-q)(q^2+1)z^{q^2-q}+(q-1)(q^2+1)z^{q^2}. 
\end{eqnarray} 
It then follows that a linear code over $\gf(q)$ has parameters $[q^2+1, 4, q^2-q]$ if 
and only if its dual has parameters $[q^2+1, q^2-3, 4]$. Ovoid codes and their duals are 
interesting due to the following: 
\begin{itemize}
\item Ovoid codes meet the Griesmer bound and are thus optimal. 
\item The duals of ovoid codes are almost-MDS. 
\item Ovoid codes and their duals can be employed to construct $3$-designs and inversive 
planes \cite{Ding18arxiv}. 
\item Ovoid codes are also the maximum minimum distance (MMD) codes \cite{FW98}. 
\end{itemize} 

Let $q=p^m$, where $p$ is a prime. Any linear code of length $n$ over $\gf(q)$ gives  
a subfield code of length $n$ over $\gf(p)$ (see Section \ref{sec-subfieldcodes}). 
The objective of this paper is to determine the parameters of the subfield codes of 
the elliptic quadric codes and Tits ovoid codes and their duals. In particular, the 
weight distributions of the subfield codes of the elliptic quadric codes and Tits ovoid 
codes are determined. As will be seen later, some of these codes are optimal. In particular, the 
duals of the subfield codes of these ovoid codes are distance-optimal. The parameters 
of the subfield codes presented in this paper are new. The optimality 
and distance optimality of these codes are the motivation of this paper.

\section{Subfield codes and their properties}\label{sec-subfieldcodes}

\subsection{Definition and basic properties}

Let $\gf(q^m)$ be a finite field with $q^m$ elements, where $q$ is a power of a prime and $m$ is a positive integer. In this section, we introduce subfield codes of linear codes and prove some basic results of subfield codes.

Given an $[n,k]$ code $\C$ over $\gf(q^m)$, we construct a new $[n, k']$ code $\C^{(q)}$ 
over $\gf(q)$ as follows. Let $G$ be a generator matrix of $\C$. Take a basis of $\gf(q^m)$ 
over $\gf(q)$. Represent each entry of $G$ as an $m \times 1$ column vector of $\gf(q)^m$ 
with respect to this basis, and replace each entry of $G$ with the corresponding $m \times 1$ column vector of $\gf(q)^m$. In this way, $G$ is modified into a $km \times n$ matrix over 
$\gf(q)$, which generates the new subfield code $\C^{(q)}$ over $\gf(q)$ with length $n$. 
By definition, the dimension $k'$ of $\C^{(q)}$ satisfies $k'\leq mk$. We will prove that 
the subfield code $\C^{(q)}$ of $\C$ is independent of the choices of both $G$ and the 
basis of $\gf(q^m)$ over $\gf(q)$. We first prove the following. 

\begin{theorem}\label{th-independent} 
For any linear code $\C$ over $\gf(q^m)$, the subfield code $\C^{(q)}$ is independent of 
the choice of the basis of $\gf(q^m)$ over $\gf(q)$ for any fixed generator matrix $G$. 
\end{theorem}

\begin{proof}
Let $\C$ be an $[n,k]$ linear code over $\gf(q^m)$. Let
$$G=\begin{bmatrix} G_1\\ G_2 \\ \vdots \\ G_k \end{bmatrix} $$
be a generator matrix of $\C$, where each $G_i$ is a $1\times n$ vector in $\gf(q^m)^n$. Choose  a basis of $\gf(q^m)$ over $\gf(q)$ and expand each element in $G_i, 1\leq i \leq k$,  under this basis as a column vector over $\gf(q)$. Then each $G_i$ is expanded as an $m\times n$ matrix $G_i^{(q)}$ over $\gf(q)$. Put
$$G^{(q)}=\begin{bmatrix} G_1^{(q)}\\ G_2^{(q)}\\ \vdots \\ G_k^{(q)} \end{bmatrix}. $$
Then $G^{(q)}$ is a generator matrix of the subfield code $\C^{(q)}$ of $\C$. Let $\mathbf{\alpha}=\{\alpha_1,\alpha_2,\cdots,\alpha_{m}\}$ and $\mathbf{\beta}=\{\beta_1,\beta_2,\cdots,\beta_{m}\}$ be any two bases of $\gf(q^m)$ over $\gf(q)$. Suppose that
$$(\alpha_1,\alpha_2,\cdots,\alpha_{m})=(\beta_1,\beta_2,\cdots,\beta_{m})T$$ where $T$ is an $m\times m$ invertible matrix over $\gf(q)$. Denote the corresponding subfield codes of $\C$ under the two bases $\mathbf{\alpha}$ and $\mathbf{\beta}$ as $\C_{\mathbf{\alpha}}^{(q)}$ and $\C_{\mathbf{\beta}}^{(q)}$, respectively. The the generator matrix $G_{\mathbf{\alpha}}^{(q)}$ of $\C_{\mathbf{\alpha}}^{(q)}$ and the generator matrix $G_{\mathbf{\beta}}^{(q)}$ of $\C_{\mathbf{\beta}}^{(q)}$ satisfy
$$G_{\mathbf{\alpha}}^{(q)}=\begin{bmatrix} T & & &\\ & T & &\\ & & \ddots & \\ & & & T \end{bmatrix}G_{\mathbf{\beta}}^{(q)}.$$
Hence, $\C_{\mathbf{\alpha}}^{(q)}$ and $\C_{\mathbf{\beta}}^{(q)}$ are the same subspace 
as $T$ is invertible. Then the desired conclusion follows.
\end{proof} 

We will prove that the subfield code $\C^{(q)}$ is also independent of the choice of 
the generator matrix $G$. To proceed in this direction, we give a trace representation of 
the subfield code. The following lemma is well-known \cite{LN} and needed later.

\begin{lemma}\label{lem-dualbasis}
Let $\{\alpha_1,\alpha_2,\cdots,\alpha_{m}\}$ be a basis of $\gf(q^m)$ over $\gf(q)$. Then there exists a unique basis $\{\beta_1,\beta_2,\cdots,\beta_{m}\}$ such that for $1\leq i,j \leq m$,
$$\tr_{q^m/q}(\alpha_i\beta_j)=\left\{
\begin{array}{rcl}
0    &   \mbox{ for }i\neq j,\\
1    &   \mbox{ for }i= j,\\
\end{array} \right. $$
i.e. the dual basis.
\end{lemma}

Lemma \ref{lem-dualbasis} directly yields the following.

\begin{lemma}\label{lem-trcerepresent}
Let $\{\alpha_1,\alpha_2,\cdots,\alpha_{m}\}$ be a basis and $\{\beta_1,\beta_2,\cdots,\beta_{m}\}$ be its dual basis of $\gf(q^m)$ over $\gf(q)$. For any $a=\sum_{i=1}^{m}a_i\alpha_i\in \gf(q^m)$ where each $a_i\in \Bbb F_q$, we then have
$$a_i=\tr_{q^m/q}(a\beta_i).$$
\end{lemma}

\begin{theorem}\label{th-genematix}
Let $\C$ be an $[n,k]$ linear code over $\gf(q^m)$ with generator matrix
$$G=\begin{bmatrix} g_{11} &g_{12} & \cdots & g_{1n}\\  g_{21} &g_{22} & \cdots & g_{2n}\\ \vdots & \vdots & \ddots & \vdots\\  g_{k1} &g_{k2} & \cdots & g_{kn}\\\end{bmatrix}.$$
Let $\{\alpha_1,\alpha_2,\cdots,\alpha_{m}\}$ be a basis of $\gf(q^m)$ over $\gf(q)$. Then the subfield code $\C^{(q)}$ of $\C$ has a generator matrix
$$G^{(q)}=\begin{bmatrix} G_1^{(q)}\\ G_2^{(q)}\\ \vdots \\ G_k^{(q)} \end{bmatrix}$$ 
where each 
$$ 
G_{i}^{(q)}=\begin{bmatrix} \tr_{q^m/q}(g_{i1}\alpha_1) &\tr_{q^m/q}(g_{i2}\alpha_1) & \cdots & \tr_{q^m/q}(g_{in}\alpha_1)\\  \tr_{q^m/q}(g_{i1}\alpha_2) &\tr_{q^m/q}(g_{i2}\alpha_2) & \cdots & \tr_{q^m/q}(g_{in}\alpha_2)\\ \vdots & \vdots & \ddots & \vdots \\ \tr_{q^m/q}(g_{i1}\alpha_m) & \tr_{q^m/q}(g_{i2}\alpha_m) &\cdots & \tr_{q^m/q}(g_{in}\alpha_m) \end{bmatrix}.$$
\end{theorem}

\begin{proof}
The desired conclusion follows from Lemma \ref{lem-trcerepresent}.
\end{proof}

With the help of Theorem \ref{th-genematix}, the trace representation of subfield codes is given in the next theorem.

\begin{theorem}\label{th-tracerepresentation}
Let $\C$ be an $[n,k]$ linear code over $\gf(q^m)$. Let $G=[g_{ij}]_{1\leq i \leq k, 1\leq j \leq n}$ be a generator matrix of $\C$. Then the trace representation of $\C^{(q)}$ is given by
$$
\C^{(q)}=\left\{\left(\tr_{q^m/q}\left(\sum_{i=1}^{k}a_ig_{i1}\right),
\cdots,\tr_{q^m/q}\left(\sum_{i=1}^{k}a_ig_{in}\right)\right):a_1,\ldots,a_k\in \gf(q^m)\right\}.
$$
\end{theorem}

\begin{proof}
 We denote $\bc=(c_1,c_2,\ldots,c_n)=\bx G^{(q)}\in \C^{(q)}$ for any
$$\bx=(x_{11},\ldots,x_{1m},\ldots,x_{k1},\ldots,x_{km})\in \gf(q)^{km}.$$
Then by Theorem \ref{th-genematix}, 
$$c_h=\sum_{i=1}^{k}\sum_{j=1}^{m}\tr_{q^m/q}(g_{ih}x_{ij}\alpha_j)=\sum_{i=1}^{k}\tr_{q^m/q}\left(g_{ih}\sum_{j=1}^{m}x_{ij}\alpha_j\right),\ 1\leq h\leq n,$$
where $\{\alpha_1,\alpha_2,\cdots,\alpha_{m}\}$ is a basis of $\gf(q^m)$ over $\gf(q)$. Note that
$$\gf(q^m)=\left\{\sum_{j=1}^{m}x_{ij}\alpha_j:(x_{i1},x_{i2},\ldots,x_{im})\in \gf(q)^m\right\}$$ with any fixed $1\leq i\leq k$. Then the desired conclusion follows.
\end{proof}

We are now ready to prove the following conclusion. 

\begin{theorem}\label{thm-2ndindepent}
The subfield code $\C^{(q)}$ of $\C$ over $\gf(q^m)$ is also independent of the choice of 
the generator matrix $G$.  
\end{theorem} 

\begin{proof}
Let $G$ and $G'$ be two generator matrices of an $[n, k]$ code $\C$ over $\gf(q^m)$. 
Then there exists a $k \times k$ invertible matrix $T$ over $\gf(q^m)$ such that 
$G'=TG$. Let $\C_{G}^{(q)}$ and $\C_{G'}^{(q)}$ denote the subfield codes with 
respect to the generator matrices $G$ and $G'$, respectively. For any 
$(a'_1, a'_2, \cdots, a'_k) \in \gf(q^m)^k$, define 
$$ 
(a_1, a_2, \cdots, a_k)=(a'_1, a'_2, \cdots, a'_k)T. 
$$  
Note that $T$ is invertible. When $(a'_1, a'_2, \cdots, a'_k)$ runs over $\gf(q^m)^k$, 
so does $(a_1, a_2, \cdots, a_k)$. It then follows from Theorem \ref{th-tracerepresentation} 
that 
\begin{eqnarray*}
\C^{(q)}_{G'}
&=& \left\{\left(\tr_{q^m/q}\left(\sum_{i=1}^{k}a'_ig'_{i1}\right),
\cdots,\tr_{q^m/q}\left(\sum_{i=1}^{k}a'_ig'_{in}\right)\right): 
a'_1,\ldots,a'_k\in \gf(q^m)\right\} \\ 
&=& \left\{\left(\tr_{q^m/q}\left(\sum_{i=1}^{k}a_ig_{i1}\right),
\cdots,\tr_{q^m/q}\left(\sum_{i=1}^{k}a_ig_{in}\right)\right):a_1,\ldots,a_k\in \gf(q^m)\right\} \\
&=& \C^{(q)}_{G}. 
\end{eqnarray*} 
This completes the proof. 
\end{proof}

Summarizing Theorems \ref{th-tracerepresentation} and \ref{thm-2ndindepent}, we conclude 
that the subfield code $\C^{(q)}$ over $\gf(q)$ of a linear code $\C$ over $\gf(q^m)$ 
is independent of the choices of both $G$ and the basis of $\gf(q^m)$ over $\gf(q)$. 
So is the dual code $\C^{(q)\perp}$.  

\subsection{Relations among $\C$, $\C^\perp$, $\C^{\perp (q)}$ and $\C^{(q)\perp}$}

Denote by $\C^{\perp}$ and $\C^{(q)\perp}$ the dual codes of $\C$ and its subfield code $\C^{(q)}$, respectively. Let $\C^{\perp (q)}$ denote the subfield code of $\C^{\perp}$. 
Since the dimensions of $\C^{\perp (q)}$ and $\C^{(q)\perp}$ vary from case to case, 
there may not be a general relation between the two codes $\C^{\perp (q)}$ and $\C^{(q)\perp}$.  

A relationship between the minimal distance of $\C^{\perp}$ and that of $\C^{(q)\perp}$ is given as follows.
\begin{theorem}\label{th-dualdistance}
Let $\C$ be an $[n,k]$ linear code over $\gf(q^m)$. Then the minimal distance $d^\perp$ of $\C^{\perp}$ and the minimal distance $d^{(q)\perp}$ of $\C^{(q)\perp}$ satisfy
$$d^{(q)\perp}\geq d^\perp.$$
\end{theorem}

\begin{proof}
Let $G=[g_{ij}]_{1\leq i \leq k, 1\leq j \leq n}$ be a generator matrix of $\C$. Let $G^{(q)}$ be a generator matrix of $\C^{(q)}$ given in Theorem \ref{th-genematix}. Then $G^{(q)}$ is also a parity-check matrix of $\C^{(q)\perp}$. This implies that there exist $b_1,b_2,\ldots,b_{d^{(q)\perp}}\in \gf(q)^*$ and integers $1\leq j_1 < j_2 <\cdots<j_{d^{(q)\perp}}\leq n$ such that
$$\sum_{h=1}^{d^{(q)\perp}}b_h\tr_{q^m/q}(g_{ij_h}\alpha_l)=\tr_{q^m/q}\left(\sum_{h=1}^{d^{(q)\perp}}b_hg_{ij_h}\alpha_l\right)=0$$
for all $1\leq i \leq k$ and $1\leq l \leq m$, where $\{\alpha_1,\alpha_2,\cdots,\alpha_{m}\}$ is a basis of $\gf(q^m)$ over $\gf(q)$. Hence
$$\sum_{l=1}^{m}\tr_{q^m/q}\left(\sum_{h=1}^{d^{(q)\perp}}b_hg_{ij_h}(u_l\alpha_l)\right)=\tr_{q^m/q}\left(\left(\sum_{h=1}^{d^{(q)\perp}}b_hg_{ij_h} \right) 
\sum_{l=1}^{m}u_l\alpha_l\right)=0$$
for all $1\leq i \leq k$ and $u_l\in \gf(q)$, $1\leq l \leq m$. Consequently,
$$\sum_{h=1}^{d^{(q)\perp}}b_hg_{ij_h}=0$$ for all $1\leq i \leq k$. Thus there exists a codeword with Hamming weight $d^{(q)\perp}$ in $\C^\perp$. Then the desired conclusion follows.
\end{proof} 

\subsection{Equivalence of subfield codes}

Two linear codes $\C_1$ and $\C_2$ are {\em permutation equivalent\index{permutation equivalent 
of codes}} if there is a permutation of coordinates which sends $\C_1$ to $\C_2$. If $\C_1$ and $\C_2$ 
are permutation equivalent, so are $\C_1^\perp$ and $\C_2^\perp$. Two permutation equivalent 
linear codes have the same dimension and weight distribution.  

A \emph{monomial matrix\index{monomial matrix}} over a field $\F$ is a square matrix having exactly one 
nonzero element of $\F$  in each row and column. A monomial matrix $M$ can be written either in 
the form $DP$ or the form $PD_1$, where $D$ and $D_1$ are diagonal matrices and $P$ is a permutation 
matrix. 

Let $\C_1$ and $\C_2$ be two linear codes of the same length over $\F$. Then $\C_1$ and $\C_2$ 
are \emph{monomially equivalent\index{monomially equivalent}} if there is a nomomial matrix over $\F$ 
such that $\C_2=\C_1M$. Monomial equivalence and permutation equivalence are precisely the same for 
binary codes. If $\C_1$ and $\C_2$ are monomially equivalent, then they have the same weight distribution. 

Let $\C$ and $\C'$ be two monomially equivalent $[n,k]$ code over $\gf(q^m)$. 
Let $G=[g_{ij}]$ and $G=[g'_{ij}]$ be two generator matrices of $\C$ and $\C'$, 
respectively. By definition, there exist a permutation $\sigma$ of the set 
$\{1, 2, \cdots, n\}$ and elements $b_1, b_2, \cdots, b_n$ in $\gf(q^m)^*$ 
such that 
$$ 
g_{ij}=b_jg'_{i\sigma(j)}
$$ 
for all $1 \leq i \leq k$ and $1 \leq j \leq n$. It then follows that 
\begin{eqnarray*}
\lefteqn{\left(\tr_{q^m/q}\left(\sum_{i=1}^{k}a_ig_{i1}\right),
\cdots,\tr_{q^m/q}\left(\sum_{i=1}^{k}a_ig_{in}\right)\right) } \\
&=&  
\left(\tr_{q^m/q}\left(b_1\left(\sum_{i=1}^{k}a_ig'_{i\sigma(1)}\right)\right),
\cdots,\tr_{q^m/q}\left(b_n\left(\sum_{i=1}^{k}a_ig_{in}\right)\right)\right).  
\end{eqnarray*} 
Then the following conclusions follow from Theorem \ref{th-tracerepresentation}: 
\begin{itemize}
\item If $\C$ and $\C'$ are permutation equivalent, so are $\C^{(q)}$ and $\C'^{(q)}$. 
\item If all $b_i \in \gf(q)^*$, then $\C^{(q)}$ and $\C'^{(q)}$ are monomially equivalent.  
\end{itemize} 
However, $\C^{(q)}$ and $\C'^{(q)}$ may not be monomially equivalent even if $\C$ and $\C'$ 
are monomially equivalent.  

\subsection{Historical information and remarks}

The subfield subcode $\C|_{\gf(q)}$ of an $[n, k]$ code over $\gf(q^m)$ is the set of codewords 
in $\C$ each of whose components is in $\C$. Hence, the dimension of the subfield subcode $\C|_{\gf(q)}$ is at most $k$. Thus, the subfield code over $\gf(q)$ and subfield subcode over $\gf(q)$ of a linear code over $\gf(q^m)$ are different codes in general. Subfield codes were considered 
in \cite{CCZ} and \cite{CCD} without using the name ``subfield codes". Subfield codes were 
defined formally in \cite[p. 5117]{MagmaHK} and a Magma function for subfield codes is implemented in the Magma package. The reader is warned that the subfield codes referred 
in \cite{Bier} and \cite{BE97} are actually subfied subcodes. These lead to a confusion.    
In view of the impact of the Magma computation system, we wish to follow the Magma definition 
of subfield codes.

While subfield subcodes have been well studied due to the Delsarte theorem \cite{Dels}, 
little has been done for subfield codes of linear codes over finite fields. The subfield 
codes of several families of linear codes were considered and distance-optimal codes were 
constructed in \cite{CCZ} and \cite{CCD}. In these two references, the basic idea is to 
consider the subfield code of a linear code over $\gf(q^m)$ with good parameters and expect 
the subfield code over $\gf(q)$ to have also good parameters. In this paper, we follow 
the same idea, and consider the subfield codes of ovoid codes which are optimal with respect 
to the Griesmer bound.

\section{Auxiliary results}\label{sec-pre}

In this section, we recall characters and some character sums over finite fields which will be needed in later sections.

Let $p$ be a prime and $q=p^m$. Let $\gf(q)$ be the finite field with $q$ elements and $\alpha$ a primitive element of $\gf(q)$. Let $\tr_{q/p}$ denote the trace function from $\gf(q)$ to $\gf(p)$ given by
$$\tr_{q/p}(x)=\sum_{i=0}^{m-1}x^{p^{i}},\ x\in \gf(q).$$ Denote $\zeta_p$ as the primitive $p$-th root of complex unity.

An \emph{additive character} of $\gf(q)$ is a function $\chi: (\gf(q),+)\rightarrow \bC^{*}$ such that
$$\chi(x+y)=\chi(x)\chi(y),\ x,y\in \gf(q),$$ where $\bC^{*}$ denotes the set of all nonzero complex numbers. For any $a\in \gf(q)$, the function
$$\chi_{a}(x)=\zeta_{p}^{\tr_{q/p}(ax)},\ x\in \gf(q),$$ defines an additive character of $\gf(q)$. In addition, $\{\chi_{a}:a\in \gf(q)\}$ is a group consisting of all the additive characters of $\gf(q)$. If $a=0$, we have $\chi_0(x)=1$ for all $x\in \gf(q)$ and $\chi_0$ is referred to as the trivial additive character of $\gf(q)$. If $a=1$, we call $\chi_1$ the canonical additive character of $\gf(q)$. Clearly, $\chi_a(x)=\chi_1(ax)$. 
The orthogonality  relation of additive characters is given by
$$\sum_{x\in \gf(q)}\chi_1(ax)=\left\{
\begin{array}{rl}
q    &   \mbox{ for }a=0,\\
0    &   \mbox{ for }a\in \gf(q)^*.
\end{array} \right. $$

Let $\gf(q)^*=\gf(q)\setminus \{0\}$. A \emph{character} $\psi$ of the multiplicative group $\gf(q)^*$ is a function from  $\gf(q)^*$  to $\bC^{*}$ such that $\psi(xy)=\psi(x)\psi(y)$ for all $(x,y)\in \gf(q)\times \gf(q)$. Define the multiplication of two characters $\psi,\psi'$ by $(\psi\psi')(x)=\psi(x)\psi'(x)$ for $x\in \gf(q)^*$. All the characters of $\gf(q)^*$ are given by
$$\psi_{j}(\alpha^k)=\zeta_{q-1}^{jk}\mbox{ for }k=0,1,\cdots,q-1,$$
where $0\leq j \leq q-2$. Then all these $\psi_j$, $0\leq j \leq q-2$, form a group under the multiplication of characters and are called \emph{multiplicative characters} of $\gf(q)$. In particular, $\psi_0$ is called the trivial multiplicative character and $\eta:=\psi_{(q-1)/2}$ is referred to as the quadratic multiplicative character of  $\gf(q)$.The orthogonality relation of multiplicative characters is given by
$$\sum_{x\in \gf(q)^*}\psi_j(x)=\left\{
\begin{array}{rl}
q-1    &   \mbox{ for }j=0,\\
0    &   \mbox{ for }j\neq 0.
\end{array} \right. $$

For an additive character $\chi$ and a multiplicative character $\psi$ of $\gf(q)$, the \emph{Gauss sum} $G(\psi, \chi)$ over $\gf(q)$ is defined by
$$G(\psi,\chi)=\sum_{x\in \gf(q)^*}\psi(x)\chi(x).$$
We call $G(\eta,\chi)$ the quadratic Gauss sum over $\gf(q)$ for nontrivial $\chi$. The value of the quadratic Gauss sum is known as follows.

\begin{lemma}\label{quadGuasssum}\cite[Th. 5.15]{LN}
Let $q=p^m$ with $p$ odd. Let $\chi$ be the canonical additive character of $\gf(q)$. Then
\begin{eqnarray*}G(\eta,\chi)&=&(-1)^{m-1}(\sqrt{-1})^{(\frac{p-1}{2})^2m}\sqrt{q}\\
 &=&\left\{
\begin{array}{lll}
(-1)^{m-1}\sqrt{q}    &   \mbox{ for }p\equiv 1\pmod{4},\\
(-1)^{m-1}(\sqrt{-1})^{m}\sqrt{q}    &   \mbox{ for }p\equiv 3\pmod{4}.
\end{array} \right. \end{eqnarray*}
\end{lemma}

Let $\chi$ be a nontrivial character of $\gf(q)$ and let $f\in \gf(q)[x]$ be a polynomial of positive degree. The character sums of the form
$$\sum_{c\in \gf(q)}\chi(f(c))$$ are referred to as \emph{Weil sums}. The problem of evaluating
such character sums explicitly is very difficult in general. In certain special cases, Weil sums can be treated (see \cite[Section 4 in Chapter 5]{LN}). If $f$ is a quadratic polynomial, the Weil sum has an interesting relationship with quadratic Gauss sums, which is described in the 
following lemma.

\begin{lemma}\label{lem-charactersum}\cite[Th. 5.33]{LN}
Let $\chi$ be a nontrivial additive character of $\gf(q)$ with $q$ odd, and let $f(x)=a_2x^2+a_1x+a_0\in \gf(q)[x]$ with $a_2\neq 0$. Then
$$\sum_{c\in \gf(q)}\chi(f(c))=\chi(a_0-a_1^2(4a_2)^{-1})\eta(a_2)G(\eta,\chi).$$
\end{lemma}

\section{The subfield codes of the elliptic quadric codes}

Let $q=p^m>2$ with $p$ a prime. Let 
$\V$ be the elliptic quadric defined by 
$$\V=\{(0,0,1,0)\}\cup \{(x,y,x^2+xy+ay^2,1):x,y\in \Bbb \gf(q)\},$$ 
where $a\in \gf(q)$ is such that the polynomial $x^2+x+a$ has no root in $\gf(q)$.
Our task in this section is to study the subfield code $\C_{\V}^{(p)}$ of the elliptic quadric 
code $\C_{\V}$. 

Let $\alpha$ be a primitive element of $\gf(q)$. 
Denote
$$f_1(x,y)=x,\ f_2(x,y)=y,\ f_3(x,y)=x^2+xy+ay^2$$ and
$$G_{x,y}=\begin{bmatrix} f_1(x,y)\\ f_2(x,y)\\ f_3(x,y) \\ 1 \end{bmatrix}_{(x,y)\in \gf(q)^2} $$
which is a $4\times q^2$ matrix over $\gf(q)$. Let $\C_{\V}$ be the linear code over $\gf(q)$ with generator matrix
$$G_{\V}=\begin{bmatrix} G_{x,y}  \begin{array}{c}
0 \\
0\\
1\\
0\end{array}
\end{bmatrix}.$$

Combining the definition of $G_{\V}$ and Theorem \ref{th-tracerepresentation} yields the
following trace representation of $\C_{\V}^{(p)}$: 
\begin{eqnarray*}
\C_{\V}^{(p)}&=&\left\{\left(\left(\tr_{q/p}(uf_1(x,y)+vf_2(x,y)+wf_3(x,y))+h\right)_{(x,y)\in \gf(q)^{2}},\tr_{q/p}(w)\right):\myatop {u,v,w \in \gf(q)}{h\in \gf(p)}\right\}\\
&=&\left\{\left(\left(\tr_{q/p}(g(x,y))+h\right)_{(x,y)\in \gf(q)^{2}},\tr_{q/p}(w)\right):\myatop {u,v,w \in \gf(q)}{h\in \gf(p)}\right\}
\end{eqnarray*}
where $g(x,y):=uf_1(x,y)+vf_2(x,y)+wf_3(x,y))=ux+vy+wx^2+wxy+way^2$.

The weight distribution of $\C_{\V}^{(p)}$ will be settled separately in the following 
two cases.

\subsection{The case $p=2$}

In the case that $p=2$ and $q=2^m >2$, the weight distribution of $\C_{\V}^{(p)}$ is documented 
in the following theorem. 

\begin{theorem}\label{th-p=2}
Let $p=2$ and $m>1$. Then $\C_{\V}^{(p)}$ is a six-weight binary linear code with parameters $[2^{2m}+1,3m+1,2^{2m-1}-2^{m-1}]$ and the weight distribution in Table \ref{tab-1}. Its dual $\C_{\V}^{(p)\perp}$ has parameters $[2^{2m}+1,2^{2m}-3m,4]$.
\end{theorem}
\begin{table}[ht]
\begin{center}
\caption{The weight distribution of $\C_{\V}^{(p)}$ for $p=2$}\label{tab-1}
\begin{tabular}{cc} \hline
Weight  &  Multiplicity   \\ \hline
$0$          &  $1$ \\
$2^{2m}$  &  $1$ \\
$2^{2m-1}$  & $2(2^{2m}-1)$ \\
$2^{2m-1}-2^{m-1}$    & $2^{2m}(2^{m-1}-1)$ \\
$2^{2m-1}+2^{m-1}$    & $2^{2m}(2^{m-1}-1)$ \\
$2^{2m-1}-2^{m-1}+1$    & $2^{3m-1}$ \\
$2^{2m-1}+2^{m-1}+1$    & $2^{3m-1}$ \\
\hline
\end{tabular}
\end{center}
\end{table}

\begin{proof}
Firstly, assume that $(u,v,w)\neq (0,0,0)$.
Denote
$$N_0(u,v,w)=\sharp\{(x,y)\in \gf(q)^2:\tr_{q/p}(g(x,y))=0\}$$ and $$N_1(u,v,w)=\sharp\{(x,y)\in \gf(q)^2:\tr_{q/p}(g(x,y))=1\}.$$ By the orthogonality relation of additive characters, we have
\begin{eqnarray}\label{eqn-1}
\nonumber 2N_0(u,v,w)&=&\sum_{(x,y)\in \gf(q)^2}\sum_{z\in \gf(2)}(-1)^{z\tr(g(x,y))}\\
&=&q^2+\sum_{(x,y)\in \gf(q)^2}(-1)^{\tr(g(x,y))}.
\end{eqnarray}
We discuss the value of $N_0(u,v,w)$ in the following cases.
\begin{enumerate}
\item If $w=0$, we have $g(x,y)=ux+vy$. Since $(u,v)\neq (0,0)$, we deduce that $N_0(u,v,w)=q^2/2=2^{2m-1}$.
\item If $w\neq 0$, we denote
$$\Delta=\sum_{(x,y)\in \gf(q)^2}(-1)^{\tr(g(x,y))}.$$ Then
\begin{eqnarray*}
\Delta^2&=&\left(\sum_{(x,y)\in \gf(q)^2}(-1)^{\tr_{2^m/2}(-g(x,y))}\right)\left(\sum_{(x_1,y_1)\in \gf(q)^2}(-1)^{\tr_{2^m/2}(g(x_1,y_1))}\right)\\
&=&\sum_{(x,y)\in \gf(q)^2}\sum_{(x_1,y_1)\in \gf(q)^2}(-1)^{\tr_{2^m/2}(g(x_1,y_1)-g(x,y))}\\
&=&\sum_{(x,y)\in \gf(q)^2}\sum_{(A,B)\in \gf(q)^2}(-1)^{\tr_{2^m/2}(g(x+A,y+B)-g(x,y))}\\
&=&q^2-\sum_{(A,B)\in \gf(q)^2\setminus \{(0,0)\}}\sum_{(x,y)\in \gf(q)^2}(-1)^{\tr_{2^m/2}(g(x+A,y+B)-g(x,y))}\\
&=&q^2-\sum_{(A,B)\in \gf(q)^2\setminus \{(0,0)\}}\sum_{(x,y)\in \gf(q)^2}(-1)^{\tr_{2^m/2}(uA+vB+wA^2+waB^2+wBx+wAy)}\\
&=&q^2-\sum_{(A,B)\in \gf(q)^2\setminus \{(0,0)\}}(-1)^{\tr_{2^m/2}(uA+vB+wA^2+waB^2)}\sum_{x\in \gf(q)}(-1)^{\tr_{2^m/2}(wBx)}\\
& &\times \sum_{y\in \gf(q)}(-1)^{\tr_{2^m/2}(wAy)}\\
&=&q^2,
\end{eqnarray*}
where we used the variable substitution $x_1=x+A,y_1=y+B$ in the third equality and the last equality holds due to the orthogonality relation of additive characters. By Equation (\ref{eqn-1}), $\Delta$ is an integer. Hence $\Delta=\pm q$ and Equation (\ref{eqn-1}) implies
$$N_0(u,v,w)=2^{2m-1}\pm 2^{m-1}.$$
\end{enumerate}
Combining the two cases above yields
\begin{eqnarray*}
N_0(u,v,w)=\left\{
\begin{array}{ll}
2^{2m-1}    &   \mbox{ for }w=0,\\
2^{2m-1}\pm 2^{m-1}    &   \mbox{ for }w\neq 0,\\
\end{array} \right.
\end{eqnarray*}
where $(u,v,w)\neq (0,0,0)$ and $N_1(u,v,w)=2^{2m}-N_0(u,v,w)$.

For any codeword $\bc(u,v,w,h):=\left(\left(\tr_{2^m/2}(g(x,y))+h\right)_{(x,y)\in \gf(2^m)^{2}},\tr_{2^m/2}(w)\right)\in \C_{\V}^{(2)}$, by the foregoing discussions we deduce that
\begin{eqnarray*}
\wt(\bc(u,v,w,h))&=&\left\{
\begin{array}{ll}
0    &   \mbox{ for }(u,v,w,h)=(0,0,0,0) \\
2^{2m}    &   \mbox{ for }(u,v,w,h)=(0,0,0,1) \\
N_1(u,v,w) & \mbox{ for }w=h=0,\ (u,v)\neq (0,0) \\
N_0(u,v,w) & \mbox{ for }w=0,\ h=1,\ (u,v)\neq (0,0) \\
N_1(u,v,w) & \mbox{ for }h=0,\ w\neq 0, \tr_{2^m/2}(w)=0,\ (u,v)\in\gf(q)^2 \\
N_0(u,v,w) & \mbox{ for }h=1,\ w\neq 0, \tr_{2^m/2}(w)=0,\ (u,v)\in \gf(q)^2 \\
N_1(u,v,w)+1 & \mbox{ for }h=0,\ \tr_{2^m/2}(w)\neq0,\ (u,v)\in \gf(q)^2 \\
N_0(u,v,w)+1 & \mbox{ for }h=1,\ \tr_{2^m/2}(w)\neq0,\ (u,v)\in \gf(q)^2 \\
\end{array} \right.\\
&=&\left\{\begin{array}{ll}
0    &   \mbox{ with 1 time},\\
2^{2m}    &   \mbox{ with 1 time},\\
2^{2m-1} & \mbox{ with }2(2^{2m}-1)\ \mbox{times},\\
2^{2m-1}+2^{m-1} & \mbox{ with }2^{2m}(2^{m-1}-1)\ \mbox{times},\\
2^{2m-1}-2^{m-1} & \mbox{ with }2^{2m}(2^{m-1}-1)\ \mbox{times},\\
2^{2m-1}+2^{m-1}+1 & \mbox{ with }2^{3m-1}\ \mbox{times},\\
2^{2m-1}-2^{m-1}+1 & \mbox{ with }2^{3m-1}\ \mbox{times},
\end{array} \right.
\end{eqnarray*}
where the frequency of each weight is very easy to derive.
Since $A_0=1$, the dimension of $\C_{\V}^{(2)}$ is $3m+1$.

Note that $\C_{\V}^{(2)\perp}$ has length $2^{2m}+1$ and dimension $2^{2m}-3m$. It follows 
from Theorems \ref{th-dualdistance} and \ref{th-dualdistanceofC} that the minimal distance $d^{(2)\perp}$ of $\C_{\V}^{(2)\perp}$ satisfies $d^{(2)\perp}\geq 4$. By the sphere-packing bound, we have
$$2^{2^{2m}+1}\geq 2^{2^{2m}-3m}\left(\sum_{i=0}^{^{\left\lfloor\frac{d^{(2)\perp}-1}{2}\right\rfloor}}\binom{2^{2m}+1}{i}\right),$$ 
which implies that $d^{(2)\perp}\leq 4$, where $\lfloor x\rfloor$ 
is the floor function. Thus $d^{(2)\perp}=4$. Then the desired conclusions follow.
\end{proof}

By the proof of Theorem \ref{th-p=2}, the weight distribution of $\C_{\V}^{(2)}$ is independent of $a$. However, this will not be true for $\C_{\V}^{(p)}$ for odd $p$.

\subsection{The case $p>2$}

In the following, we investigate the weight distributions of $\C_{\V}^{(p)}$ for $p>2$.
We first present some lemmas below.

\begin{lemma}\label{lem-eta}
Let $q$ be odd and $\eta$ the quadratic multiplicative character of $\gf(q)$. Let $x^2+x+a$ be irreducible over $\gf(q)$. Then $\eta(a-4^{-1})=(-1)^{(q+1)/2}$.
\end{lemma}

\begin{proof}
Let $\alpha$ be a primitive element of $\gf(q)$. It is easily seen that
\begin{eqnarray*}\eta(-1)=\left\{\begin{array}{ll}
-1 & \mbox{ for }q\equiv 3\pmod{4},\\
1 & \mbox{ for }q\equiv 1\pmod{4}.
\end{array} \right.
\end{eqnarray*}

For $q\equiv 3\pmod{4}$, suppose that $\eta(a-4^{-1})=-1$. Then $a-4^{-1}=\alpha^{2j+1}$ for some $0\leq j\leq\frac{q-3}{2}$. The  discriminant of $x^2+x+a=0$ equals $-\alpha^{2j+1}$ which is a square in $\gf(q)$ as $\eta(-1)=-1$. This contradicts with the fact that $x^2+x+a$ is irreducible. Hence,  $\eta(a-4^{-1})=1$.

For $q\equiv 1\pmod{4}$, suppose that $\eta(a-4^{-1})=1$. Then $a-4^{-1}=\alpha^{2j}$ for some $0\leq j\leq\frac{q-3}{2}$. The  discriminant of $x^2+x+a=0$ equals $-\alpha^{2j}$ which is a square in $\gf(q)$ as $\eta(-1)=1$. This contradicts with the fact that $x^2+x+a$ is irreducible. Hence,  $\eta(a-4^{-1})=-1$. 
Then the desired conclusions follow.
\end{proof}

\begin{lemma}\label{lem-eta-reducible}
Let $q$ be odd and $\eta$ the quadratic multiplicative character of $\gf(q)$. Let $x^2+x+a$ be reducible over $\gf(q)$ and $a\neq 4^{-1}$. Then $\eta(a-4^{-1})=(-1)^{(q-1)/2}$.
\end{lemma}

\begin{proof}
Since $x^2+x+a$ is reducible over $\gf(q)$ and $a\neq 4^{-1}$, we have $a=4^{-1}-b^2$ for some $b\in \gf(q)^*$. Hence, $\eta(a-4^{-1})=\eta(-b^2)=\eta(-1)$. Recall that
\begin{eqnarray*}\eta(-1)=\left\{\begin{array}{ll}
-1 & \mbox{ for }q\equiv 3\pmod{4},\\
1 & \mbox{ for }q\equiv 1\pmod{4}.
\end{array} \right.
\end{eqnarray*}
Then the desired conclusion follows.
\end{proof}

\begin{lemma}\label{lem-odd}
Let $q=p^m$ with $p$ odd. Then $(\frac{p-1}{2})^{2}m+\frac{q+1}{2}$ is an odd integer.
\end{lemma}

\begin{proof}
Note that $(\frac{p-1}{2})^{2}m+\frac{q+1}{2}=\frac{p^2m-2pm+2q+m+2}{4}$. Denote $s=p^2m-2pm+2q+m+2$. We discuss the value of $s$ in two cases.
\begin{enumerate}
\item Let $p\equiv 1\pmod{4}$. Assume that $p=4t+1$ for some positive integer $t$. Then
\begin{eqnarray*}
s&=&(4t+1)^2m-2(4t+1)m+2(4t+1)^{m}+m+2\\
&=&16t^2m+2(4t+1)^{m}+2\\
&\equiv& 4\pmod{8}.
\end{eqnarray*}
\item Let $p\equiv 3\pmod{4}$. Assume that $p=4t+3$ for some nonnegative integer $t$. Then
\begin{eqnarray*}
s&=&(4t+3)^2m-2(4t+3)m+2(4t+3)^{m}+m+2\\
&=&16t^2m+16tm+2(4t+3)^{m}+4m+2\\
&\equiv& 4\pmod{8}.
\end{eqnarray*}
\end{enumerate}
Then the desired conclusion follows.
\end{proof}

\begin{lemma}\label{lem-even}
Let $q=p^m$ with $p$ odd. Then $(\frac{p-1}{2})^{2}m+\frac{q-1}{2}$ is an even integer.
\end{lemma}
\begin{proof}
The proof is similar to that of Lemma \ref{lem-odd} and is omitted.
\end{proof}

The weight distributions of $\C_{\V}^{(p)}$ are given in three cases according to different choices of $a$ as follows.
\begin{theorem}\label{th-p>2-case1}
Let $p>2$, $m>1$ and $a\in \gf(q)$ such that $x^2+x+a$ has no root in $\gf(q)$. Then $\C_{\V}^{(p)}$ is a six-weight $p$-ary linear code with parameters $[p^{2m}+1,3m+1,p^{2m-1}(p-1)-p^{m-1}]$ and the weight distribution in Table \ref{tab-2}. Its dual $\C_{\V}^{(p)\perp}$ has parameters $[p^{2m}+1,p^{2m}-3m,4]$.
\end{theorem}

\begin{table}[ht]
\begin{center}
\caption{The weight distribution of $\C_{\V}^{(p)}$ for $p>2$ and irreducible $x^2+x+a$}\label{tab-2}
\begin{tabular}{cc} \hline
Weight  &  Multiplicity   \\ \hline
$0$          &  $1$ \\
$p^{2m}$  &  $p-1$ \\
$p^{2m-1}(p-1)$  & $p(p^{2m}-1)$ \\
$(p^{2m-1}+p^{m-1})(p-1)$   &  $p^{2m}(p^{m-1}-1)$\\
$p^{2m-1}(p-1)-p^{m-1}$  &   $p^{2m}(p^{m-1}-1)(p-1)$\\
$(p^{2m-1}+p^{m-1})(p-1)+1$   &  $p^{3m-1}(p-1)$\\
$p^{2m-1}(p-1)-p^{m-1}+1$  &   $p^{3m-1}(p-1)^2$\\
\hline
\end{tabular}
\end{center}
\end{table}

\begin{proof}
Let $\chi$ be the canonical additive character of $\gf(q)$.
Denote
$$N(u,v,w,h)=\sharp\{(x,y)\in \gf(q)^2:\tr_{q/p}(g(x,y))+h=0\}.$$ By the orthogonality relation of additive characters, we have
\begin{eqnarray}\label{eqn-2}
\nonumber pN(u,v,w,h)&=&\sum_{(x,y)\in \gf(q)^2}\sum_{z\in \gf(p)}\zeta_{p}^{z(\tr_{q/p}(g(x,y))+h)}\\
\nonumber&=&q^2+\sum_{z\in \gf(p)^*}\zeta_{p}^{zh}\sum_{(x,y)\in \gf(q)^2}\chi(zg(x,y))\\
&=&q^2+\Omega,
\end{eqnarray}
where $$\Omega:=\sum_{z\in \gf(p)^*}\zeta_{p}^{zh}\sum_{(x,y)\in \gf(q)^2}\chi(zg(x,y)).$$ Recall that $g(x,y)=uf_1(x,y)+vf_2(x,y)+wf_3(x,y))=ux+vy+wx^2+wxy+way^2$. Then we have
\begin{eqnarray}\label{eqn-3}
\nonumber\Omega&=&\sum_{z\in \gf(p)^*}\zeta_{p}^{zh}\sum_{(x,y)\in \gf(q)^2}\chi(z(ux+vy+wx^2+wxy+way^2))\\
&=&\sum_{z\in \gf(p)^*}\zeta_{p}^{zh}\sum_{y\in \gf(q)}\chi(zway^2+zvy)\sum_{x\in \gf(q)}\chi\left(zwx^2+(zu+zwy)x\right).
\end{eqnarray}
If $(u,v,w)\neq (0,0,0)$, we discuss the value of $\Omega$ in the following cases.
\begin{enumerate}
\item Assume that $w\neq 0$. Using Lemma \ref{lem-charactersum}, we get that 
\begin{eqnarray}\label{eqn-4}
\lefteqn{\sum_{x\in \gf(q)}\chi\left(zwx^2+(zu+zwy)x\right)} \nonumber \\
\nonumber&=&\chi\left(-(zu+zwy)^2(4zw)^{-1}\right)\eta(zw)G(\eta,\chi)\\
&=&\chi(-4^{-1}zwy^2-2^{-1}uzy-z(4w)^{-1}u^2)\eta(zw)G(\eta,\chi).
\end{eqnarray}
Note that $a-4^{-1}\neq 0$ as $x^2+x+a$ is irreducible over $\gf(q)$. Combining Equations (\ref{eqn-3}) and (\ref{eqn-4}) yields that 
\begin{eqnarray*}
\Omega&=&G(\eta,\chi)\sum_{z\in \gf(p)^*}\zeta_{p}^{zh}\eta(zw)\sum_{y\in \gf(q)}\chi\left((zwa-4^{-1}zw)y^2+(zv-2^{-1}uz)y-z(4w)^{-1}u^2\right)\\
&=&G(\eta,\chi)^2\sum_{z\in \gf(p)^*}\zeta_{p}^{zh}\eta(zw)\chi\left(-z(4w)^{-1}u^2-(zv-2^{-1}uz)^2(4zwa-zw)^{-1}\right)\\
& &\times \eta(zwa-4^{-1}zw)\\
&=&G(\eta,\chi)^2\sum_{z\in \gf(p)^*}\zeta_{p}^{zh}\eta(zw)^2\chi\left(-z(4w)^{-1}u^2-(zv-2^{-1}uz)^2(4zwa-zw)^{-1}\right)\eta(a-4^{-1})\\
&=&G(\eta,\chi)^2\eta(a-4^{-1})\sum_{z\in \gf(p)^*}\zeta_{p}^{zh}\chi\left(-zw^{-1}(4^{-1}u^2+(v-2^{-1}u)^2(4a-1)^{-1})\right)\\
&=&\left\{
\begin{array}{ll}
G(\eta,\chi)^2\eta(a-4^{-1})\sum_{z\in \gf(p)^*}\zeta_{p}^{zh}    &   \mbox{ for }(u,v)=(0,0) \\
G(\eta,\chi)^2\eta(a-4^{-1})\sum_{z\in \gf(p)^*}\zeta_{p}^{zh}\chi(cz)    &   \mbox{ for }(u,v)\neq(0,0) \\
\end{array} \right.\\
&=&\left\{
\begin{array}{ll}
(-1)^{(\frac{p-1}{2})^{2}m+\frac{q+1}{2}}q\sum_{z\in \gf(p)^*}\zeta_{p}^{zh}    &   \mbox{ for }(u,v)=(0,0) \\
(-1)^{(\frac{p-1}{2})^{2}m+\frac{q+1}{2}}q\sum_{z\in \gf(p)^*}\zeta_{p}^{zh}\chi(cz)    &   \mbox{ for }(u,v)\neq(0,0) \\
\end{array} \right.\\
&=&\left\{
\begin{array}{ll}
-q\sum_{z\in \gf(p)^*}\zeta_{p}^{zh}    &   \mbox{ for }(u,v)=(0,0),\\
-q\sum_{z\in \gf(p)^*}\zeta_{p}^{zh}\chi(cz)    &   \mbox{ for }(u,v)\neq(0,0),\\
\end{array} \right.
\end{eqnarray*}
where $c:=-w^{-1}(4^{-1}u^2+(v-2^{-1}u)^2(4a-1)^{-1})$ for $(u,v)\neq(0,0)$, the second equality holds due to Lemma \ref{lem-charactersum} and the last two equalities hold by Lemmas \ref{quadGuasssum}, \ref{lem-eta} and \ref{lem-odd}.  Then we further have
\begin{eqnarray*}
\Omega&=&\left\{
\begin{array}{ll}
-q(p-1)   &   \mbox{ for }(u,v)=(0,0),\ h=0 \\
q   &   \mbox{ for }(u,v)=(0,0),\ h\neq0 \\
-q\sum_{z\in \gf(p)^*}\zeta_{p}^{\tr_{q/p}(c)z}    &   \mbox{ for }(u,v)\neq(0,0),\ h=0 \\
-q\sum_{z\in \gf(p)^*}\zeta_{p}^{(h+\tr_{q/p}(c))z}    &   \mbox{ for }(u,v)\neq(0,0),\ h\neq 0 \\
\end{array} \right.\\
&=&\left\{
\begin{array}{ll}
-q(p-1)   &   \mbox{ for }(u,v)=(0,0),\ h=0 \\
q   &   \mbox{ for }(u,v)=(0,0),\ h\neq 0 \\
-q(p-1)    &   \mbox{ for }(u,v)\neq(0,0),\ h=0,\ \tr_{q/p}(c)=0 \\
q    &   \mbox{ for }(u,v)\neq(0,0),\ h=0,\ \tr_{q/p}(c)\neq0 \\
-q(p-1)  &   \mbox{ for }(u,v)\neq(0,0),\ h\neq 0,\ h+\tr_{q/p}(c)=0 \\
q    &   \mbox{ for }(u,v)\neq(0,0),\ h\neq 0,\ h+\tr_{q/p}(c)\neq0 \\
\end{array} \right.\\
&=&\left\{
\begin{array}{ll}
-q(p-1)   &  \myatop{ \mbox{ for $(u,v)=(0,0),\ h=0$ or}}{\mbox{$(u,v)\neq(0,0),\ h+\tr_{q/p}(c)=0$},}\\
q   &   \myatop{\mbox{ for $(u,v)=(0,0),\ h\neq0$ or}}{\mbox{$(u,v)\neq(0,0),\ h+\tr_{q/p}(c)\neq0$},}\\
\end{array} \right.
\end{eqnarray*}
when $(u,v,w,h)$ runs through $\gf(q)\times \gf(q)\times \gf(q)^*\times \gf(p)$.
\item Assume that $w=0$ and $(u,v)\neq (0,0)$. Then Equation (\ref{eqn-3}) implies
\begin{eqnarray*}
\Omega&=&\sum_{z\in \gf(p)^*}\zeta_{p}^{zh}\sum_{(x,y)\in \gf(q)^2}\chi(z(ux+vy))\\
&=&\sum_{z\in \gf(p)^*}\zeta_{p}^{zh}\sum_{x\in \gf(q)^2}\chi(zux)\sum_{x\in \gf(q)^2}\chi(zvy)\\
&=&0
\end{eqnarray*}
as $(u,v)\neq (0,0)$.
\end{enumerate}
By Equation (\ref{eqn-2}) and the discussions above, we deduce that
\begin{eqnarray*}
N(u,v,w,h)=\left\{
\begin{array}{ll}
p^{2m} & \mbox{ for }(u,v,w,h)=(0,0,0,0),\\
0 & \mbox{ for }(u,v,w)=(0,0,0,0),\ h\neq 0,\\
p^{2m-1} & \mbox{ for $(u,v)\neq (0,0)$ and $w=0$},\\
p^{2m-1}-p^{m-1}(p-1)   &  \myatop{ \mbox{ for $(u,v)=(0,0),\ h=0,\ w\neq 0$ or}}{\mbox{$(u,v)\neq(0,0),\ h+\tr_{q/p}(c)=0,\ w\neq 0$},}\\
p^{2m-1}+p^{m-1}  &   \myatop{\mbox{ for $(u,v)=(0,0),\ h\neq0,\ w\neq 0$ or}}{\mbox{$(u,v)\neq(0,0),\ h+\tr_{q/p}(c)\neq0,\ w\neq 0$},}
\end{array} \right.
\end{eqnarray*}
where $c\in \gf(q)^*$  is defined as above.

For any codeword $$\bc(u,v,w,h):=\left(\left(\tr_{p^m/p}(g(x,y))+h\right)_{(x,y)\in \gf(p^m)^{2}},\tr_{p^m/p}(w)\right)\in \C_{\V}^{(p)},$$ by the discussions above  
we deduce that
\begin{eqnarray*}
\lefteqn{\wt(\bc(u,v,w,h)) } \\
&=&\left\{
\begin{array}{ll}
0 & \mbox{ for }(u,v,w,h)=(0,0,0,0) \\
p^{2m} & \mbox{ for }(u,v,w)=(0,0,0,0),\ h\neq 0 \\
p^{2m-1}(p-1) & \mbox{ for $(u,v)\neq (0,0)$ and $w=0$} \\
(p^{2m-1}+p^{m-1})(p-1)   &  \myatop{ \mbox{ for $(u,v)=(0,0),\ h=0,\ w\neq 0,\ \tr_{p^m/p}(w)=0$ or}}{\mbox{$(u,v)\neq(0,0),\ h+\tr_{q/p}(c)=0,\ w\neq 0,\ \tr_{p^m/p}(w)=0$}} \\
p^{2m-1}(p-1)-p^{m-1}  &   \myatop{\mbox{ for $(u,v)=(0,0),\ h\neq0,\ w\neq 0,\ \tr_{p^m/p}(w)=0$ or}}{\mbox{$(u,v)\neq(0,0),\ h+\tr_{q/p}(c)\neq0,\ w\neq 0,\ \tr_{p^m/p}(w)=0$}} \\
(p^{2m-1}+p^{m-1})(p-1)+1   &  \myatop{ \mbox{ for $(u,v)=(0,0),\ h=0,\ \tr_{p^m/p}(w)\neq0$ or}}{\mbox{$(u,v)\neq(0,0),\ h+\tr_{q/p}(c)=0,\ \tr_{p^m/p}(w)\neq0$}} \\
p^{2m-1}(p-1)-p^{m-1}+1  &   \myatop{\mbox{ for $(u,v)=(0,0),\ h\neq0,\ \tr_{p^m/p}(w)\neq0$ or}}{\mbox{$(u,v)\neq(0,0),\ h+\tr_{q/p}(c)\neq0,\ \tr_{p^m/p}(w)\neq0$}} \\
\end{array} \right.\\
&=&\left\{
\begin{array}{ll}
0 & \mbox{ with 1 time},\\
p^{2m} & \mbox{ with $p-1$ times},\\
p^{2m-1}(p-1) & \mbox{ with $p(p^{2m}-1)$ times},\\
(p^{2m-1}+p^{m-1})(p-1)   &  \mbox{ with $p^{2m}(p^{m-1}-1)$ times,}\\
p^{2m-1}(p-1)-p^{m-1}  &   \mbox{ with $p^{2m}(p^{m-1}-1)(p-1)$ times,}\\
(p^{2m-1}+p^{m-1})(p-1)+1   &  \mbox{ with $p^{3m-1}(p-1)$ times,}\\
p^{2m-1}(p-1)-p^{m-1}+1  &   \mbox{ with $p^{3m-1}(p-1)^2$ times,}\\
\end{array} \right.
\end{eqnarray*}
when $(u,v,w,h)$ runs through $\gf(q)\times \gf(q)\times \gf(q)\times \gf(p)$. Note that the dimension of $\C_{\V}^{(p)}$ is $3m+1$ as $A_0=1$.

Note that $\C_{\V}^{(p)\perp}$ is of length $p^{2m}+1$ and dimension $p^{2m}-3m$. It follows from Theorems \ref{th-dualdistance} and \ref{th-dualdistanceofC} that the minimal distance $d^{(p)\perp}$ of $\C_{\V}^{(p)\perp}$ satisfies $d^{(p)\perp}\geq 4$. By the sphere-packing bound, we have
$$p^{p^{2m}+1}\geq p^{p^{2m}-3m}\left(\sum_{i=0}^{^{\left\lfloor\frac{d^{(p)\perp}-1}{2}\right\rfloor}}(p-1)^i\binom{p^{2m}+1}{i}\right),$$ 
which implies that $d^{(p)\perp}\leq 4$, where $\lfloor x\rfloor$ is the floor function. Thus $d^{(p)\perp}=4$. 
The proof is now completed.
\end{proof}

\begin{theorem}\label{th-p>2-case2}
Let $p>2$, $m>1$ and $a\in \gf(q)$ such that $x^2+x+a$ is reducible over $\gf(q)$ and $a\neq 4^{-1}$. Then $\C_{\V}^{(p)}$ is a six-weight $p$-ary linear code with parameters $[p^{2m}+1,3m+1,(p^{2m-1}-p^{m-1})(p-1)]$ and the weight distribution in Table \ref{tab-3}.
\end{theorem}

\begin{table}[ht]
\begin{center}
\caption{The weight distribution of $\C_{\V}^{(p)}$ for $p>2$ and reducible $x^2+x+a$}\label{tab-3}
\begin{tabular}{cc} \hline
Weight  &  Multiplicity   \\ \hline
$0$          &  $1$ \\
$p^{2m}$  &  $p-1$ \\
$p^{2m-1}(p-1)$  & $p(p^{2m}-1)$ \\
$(p^{2m-1}-p^{m-1})(p-1)$   &  $p^{2m}(p^{m-1}-1)$\\
$p^{2m-1}(p-1)+p^{m-1}$  &   $p^{2m}(p^{m-1}-1)(p-1)$\\
$(p^{2m-1}-p^{m-1})(p-1)+1$   &  $p^{3m-1}(p-1)$\\
$p^{2m-1}(p-1)+p^{m-1}+1$  &   $p^{3m-1}(p-1)^2$\\
\hline
\end{tabular}
\end{center}
\end{table}
\begin{proof}
The proof of this theorem and that of Theorem \ref{th-p>2-case1} are almost exactly the same except for using Lemmas \ref{lem-eta-reducible} and \ref{lem-even} instead of Lemmas \ref{lem-eta} and \ref{lem-odd}. We omit the details of the proof here.
\end{proof}

\begin{lemma}\label{lem-equalities}
Let $q=p^m$ with $p$ an odd prime. Then the following statements hold.
\begin{enumerate}
\item \begin{eqnarray*}
\lefteqn{\sharp \{w\in \gf(q)^*:\eta(w)=1\mbox{ and }\tr_{q/p}(w)=0\} } \\
&=&\left\{
\begin{array}{ll}
\frac{p^{m-1}-1-(p-1)p^{\frac{m-2}{2}}(\sqrt{-1})^{\frac{(p-1)m}{2}}}{2} & \mbox{ for even $m$,} \\
\frac{p^{m-1}-1}{2} & \mbox{ for odd $m$.} 
\end{array}\right.
\end{eqnarray*}
\item \begin{eqnarray*}
\lefteqn{ \sharp \{w\in \gf(q)^*:\eta(w)=1\mbox{ and }\tr_{q/p}(w) \neq 0\} }\\
&=&\left\{
\begin{array}{ll}
\frac{(p-1)(p^{m-1}+p^{\frac{m-2}{2}}(\sqrt{-1})^{\frac{(p-1)m}{2}})}{2} & \mbox{ for even $m$},\\
\frac{p^{m-1}(p-1)}{2} & \mbox{ for odd $m$}.
\end{array}\right.
\end{eqnarray*}
\item \begin{eqnarray*}
\lefteqn{ \sharp \{w\in \gf(q)^*:\eta(w)=-1\mbox{ and }\tr_{q/p}(w)=0\} } \\
&=&\left\{
\begin{array}{ll}
\frac{p^{m-1}-1+(p-1)p^{\frac{m-2}{2}}(\sqrt{-1})^{\frac{(p-1)m}{2}}}{2} & \mbox{ for even $m$},\\
\frac{p^{m-1}-1}{2} & \mbox{ for odd $m$}.
\end{array}\right.
\end{eqnarray*}
\item \begin{eqnarray*}
\lefteqn{ \sharp \{w\in \gf(q)^*:\eta(w)=-1\mbox{ and }\tr_{q/p}(w)\neq0\} } \\
&=&\left\{
\begin{array}{ll}
\frac{(p-1)(p^{m-1}-p^{\frac{m-2}{2}}(\sqrt{-1})^{\frac{(p-1)m}{2}})}{2} & \mbox{ for even $m$},\\
\frac{p^{m-1}(p-1)}{2} & \mbox{ for odd $m$}.
\end{array}\right.
\end{eqnarray*}
\end{enumerate}
\end{lemma}

\begin{proof}
We only prove the first equality as the others follow directly.
Let $\chi$ be the canonical additive character and $\alpha$ a primitive element of $\gf(q)$. Let $C_0$ be the cyclic group generated by $\alpha^2$.
Denote $N(w)=\sharp \{w\in \gf(q)^*:\eta(w)=1\mbox{ and }\tr_{q/p}(w)=0\}$. By the orthogonality  relation of additive characters and Lemmas \ref{quadGuasssum} and \ref{lem-charactersum}, we obtain that 
\begin{eqnarray*}
N(w)&=&\frac{1}{p}\sum_{z\in \gf(p)}\sum_{w\in C_0}\chi(zw)\\
&=&\frac{1}{2p}\sum_{z\in \gf(p)}\sum_{w\in \gf(q)^*}\chi(zw^2)\\
&=&\frac{q-p}{2p}+\frac{1}{2p}\sum_{z\in \gf(p)^*}\sum_{w\in \gf(q)}\chi(zw^2)\\
&=&\frac{q-p}{2p}+\frac{1}{2p}G(\eta,\chi)\sum_{z\in \gf(p)^*}\eta(z)\\
&=&\left\{
\begin{array}{ll}
\frac{q-p}{2p}+\frac{p-1}{2p}G(\eta,\chi) & \mbox{ for even $m$} \\
\frac{q-p}{2p} & \mbox{ for odd $m$} 
\end{array}\right.\\
&=&\left\{
\begin{array}{ll}
\frac{p^{m-1}-1-(p-1)p^{\frac{m-2}{2}}(\sqrt{-1})^{\frac{(p-1)m}{2}}}{2} & \mbox{ for even $m$},\\
\frac{p^{m-1}-1}{2} & \mbox{ for odd $m$}, 
\end{array}\right.
\end{eqnarray*}
where the fifth equality comes from the orthogonality of the multiplicative characters. 
\end{proof}

\begin{theorem}\label{th-p>2-case3}
Let $p>2$, $m>1$ and $a=4^{-1}$. Then $\C_{\V}^{(p)}$ is a $p$-ary $[p^{2m}+1,3m+1]$  linear code with weight distributions in Tables \ref{tab-4} and \ref{tab-5} for even $m$ and  odd $m$, respectively.
\end{theorem}

\begin{table}[ht]
\begin{center}
\caption{The weight distribution of $\C_{\V}^{(p)}$ for $p>2$, even $m$ and $a=4^{-1}$}\label{tab-4}
\begin{tabular}{cc} \hline
Weight  &  Multiplicity   \\ \hline
0 & 1\\
$p^{2m}$ & $p-1$\\
$p^{2m-1}(p-1)$ & $(p^m-1)(p^{2m}+p)$\\
$p^{2m-1}(p-1)+1$ & $p^{2m}(p^m-1)(p-1)$\\
$p^{m-1}(p-1)\left(p^m+p^{\frac{m}{2}}(\sqrt{-1})^{\frac{p-1}{2}m}\right)$ & $\frac{p^{2m-1}-p^m-(p-1)p^{\frac{3m-2}{2}}(\sqrt{-1})^{\frac{(p-1)m}{2}}}{2}$\\
$p^{m-1}(p-1)\left(p^m+p^{\frac{m}{2}}(\sqrt{-1})^{\frac{p-1}{2}m}\right)+1$ & $\frac{(p-1)(p^{2m-1}+p^{\frac{3m-2}{2}}(\sqrt{-1})^{\frac{(p-1)m}{2}})}{2}$\\
$p^{m-1}\left(p^m(p-1)-p^{\frac{m}{2}}(\sqrt{-1})^{\frac{p-1}{2}m}\right)$ & $\frac{(p-1)(p^{2m-1}-p^m-(p-1)p^{\frac{3m-2}{2}}(\sqrt{-1})^{\frac{(p-1)m}{2}})}{2}$\\
$p^{m-1}\left(p^m(p-1)-p^{\frac{m}{2}}(\sqrt{-1})^{\frac{p-1}{2}m}\right)+1$ & $\frac{(p-1)^2(p^{2m-1}+p^{\frac{3m-2}{2}}(\sqrt{-1})^{\frac{(p-1)m}{2}})}{2}$\\
$p^{m-1}(p-1)\left(p^m-p^{\frac{m}{2}}(\sqrt{-1})^{\frac{p-1}{2}m}\right)$ & $\frac{p^{2m-1}-p^m+(p-1)p^{\frac{3m-2}{2}}(\sqrt{-1})^{\frac{(p-1)m}{2}}}{2}$\\
$p^{m-1}(p-1)\left(p^m-p^{\frac{m}{2}}(\sqrt{-1})^{\frac{p-1}{2}m}\right)+1$ & $\frac{(p-1)(p^{2m-1}-p^{\frac{3m-2}{2}}(\sqrt{-1})^{\frac{(p-1)m}{2}})}{2}$\\
$p^{m-1}\left(p^m(p-1)+p^{\frac{m}{2}}(\sqrt{-1})^{\frac{p-1}{2}m}\right)$ & $\frac{(p-1)(p^{2m-1}-p^m+(p-1)p^{\frac{3m-2}{2}}(\sqrt{-1})^{\frac{(p-1)m}{2}})}{2}$\\
$p^{m-1}\left(p^m(p-1)+p^{\frac{m}{2}}(\sqrt{-1})^{\frac{p-1}{2}m}\right)+1$ & $\frac{(p-1)^2(p^{2m-1}-p^{\frac{3m-2}{2}}(\sqrt{-1})^{\frac{(p-1)m}{2}})}{2}$\\
\hline
\end{tabular}
\end{center}
\end{table}

\begin{table}[ht]
\begin{center}
\caption{The weight distribution of $\C_{\V}^{(p)}$ for $p>2$, odd $m$ and $a=4^{-1}$}\label{tab-5}
\begin{tabular}{cc} \hline
Weight  &  Multiplicity   \\ \hline
0 & 1\\
$p^{2m}$ & $p-1$\\
$p^{2m-1}(p-1)$ & $p^m(p^{m-1}-1)(p^{m+1}-p+1)+p(p^{2m}-1)$\\
$p^{2m-1}(p-1)+1$ & $p^{2m-1}(p-1)(p^{m+1}-p+1)$\\
$p^{\frac{3m-1}{2}}\left(p^{\frac{m-1}{2}}(p-1)-(\sqrt{-1})^\frac{(p-1)(m+1)}{2}\right)$ & $\frac{p^m(p^{m-1}-1)(p-1)}{2}$\\
$p^{\frac{3m-1}{2}}\left(p^{\frac{m-1}{2}}(p-1)-(\sqrt{-1})^\frac{(p-1)(m+1)}{2}\right)+1$ & $\frac{p^{2m-1}(p-1)^2}{2}$\\
$p^{\frac{3m-1}{2}}\left(p^{\frac{m-1}{2}}(p-1)+(\sqrt{-1})^\frac{(p-1)(m+1)}{2}\right)$ & $\frac{p^m(p^{m-1}-1)(p-1)}{2}$\\
$p^{\frac{3m-1}{2}}\left(p^{\frac{m-1}{2}}(p-1)+(\sqrt{-1})^\frac{(p-1)(m+1)}{2}\right)+1$ & $\frac{p^{2m-1}(p-1)^2}{2}$\\
\hline
\end{tabular}
\end{center}
\end{table}

\begin{proof}
We follow the notation in the proof of Theorem \ref{th-p>2-case1}, where 
$$N(u,v,w,h)=\sharp\{(x,y)\in \gf(q)^2:\tr_{q/p}(g(x,y))+h=0\}.$$
Let $a=4^{-1}$.

If $w\neq 0$, by Equations (\ref{eqn-2}), (\ref{eqn-3}) and (\ref{eqn-4}) and Lemma \ref{lem-charactersum} we have
$$pN(u,v,w,h)=q^2+\Omega,$$ 
where 
\begin{eqnarray*}
\Omega&=&G(\eta,\chi)\sum_{z\in \gf(p)^*}\zeta_{p}^{zh}\eta(zw)\sum_{y\in \gf(q)}\chi\left((zwa-4^{-1}zw)y^2+(zv-2^{-1}uz)y-z(4w)^{-1}u^2\right)\\
&=&G(\eta,\chi)\sum_{z\in \gf(p)^*}\zeta_{p}^{zh}\eta(zw)\chi(-z(4w)^{-1}u^2)\sum_{y\in \gf(q)}\chi\left((zv-2^{-1}uz)y\right)\\
&=&\left\{
\begin{array}{ll}
0 & \mbox{ for }v\neq 2^{-1}u \\
G(\eta,\chi)q\sum\limits_{z\in \gf(p)^*}\zeta_{p}^{zh}\eta(zw)\chi(-z(4w)^{-1}u^2) & \mbox{ for }v=2^{-1}u \\
\end{array}\right.\\
&=&\left\{
\begin{array}{ll}
0 & \mbox{ for }v\neq 2^{-1}u,\\
G(\eta,\chi)q\eta(w)\sum\limits_{z\in \gf(p)^*}\zeta_{p}^{(h-\tr_{q/p}(w^{-1}v^2))z}\eta(z) & \mbox{ for }v=2^{-1}u. \\
\end{array}\right.
\end{eqnarray*}
When $m$ is even, we have $\eta(z)=1$ for $z\in \gf(p)^*$. When $m$ is odd, $\eta(z)=\eta'(z)$ for $z\in \gf(p)^*$, where $\eta'$ denotes the quadratic multiplicative character of $\gf(p)$. Let $\chi'$ denote the canonical additive character of $\gf(p)$. Then we have the following.
\begin{enumerate}
\item When $m$ is even, we deduce that
\begin{eqnarray*}
\Omega&=&\left\{
\begin{array}{ll}
0 & \mbox{ for }v\neq 2^{-1}u,\\
G(\eta,\chi)q(p-1) & \mbox{ for }v=2^{-1}u,\ h=\tr_{q/p}(w^{-1}v^2),\ \eta(w)=1,\\
-G(\eta,\chi)q & \mbox{ for }v=2^{-1}u,\ h\neq\tr_{q/p}(w^{-1}v^2),\ \eta(w)=1,\\
-G(\eta,\chi)q(p-1) & \mbox{ for }v=2^{-1}u,\ h=\tr_{q/p}(w^{-1}v^2),\ \eta(w)=-1,\\
G(\eta,\chi)q & \mbox{ for }v=2^{-1}u,\ h\neq\tr_{q/p}(w^{-1}v^2),\ \eta(w)=-1.
\end{array}\right.
\end{eqnarray*}
\item When $m$ is odd, we deduce that
\begin{eqnarray*}
\Omega&=&\left\{
\begin{array}{ll}
0 & \mbox{ for }v\neq 2^{-1}u \\
G(\eta,\chi)q\eta(w)\sum\limits_{z\in \gf(p)^*}\chi'((h-\tr_{q/p}(w^{-1}v^2))z)\eta'(z) & \mbox{ for }v=2^{-1}u \\
\end{array}\right.\\
&=&\left\{
\begin{array}{ll}
0 & \mbox{ for }v\neq 2^{-1}u \\
0 & \mbox{ for }v=2^{-1}u,\ h=\tr_{q/p}(w^{-1}v^2) \\
G(\eta,\chi)G(\eta',\chi')q\eta(w)\eta'(h-\tr_{q/p}(w^{-1}v^2) & \mbox{ for }v=2^{-1}u,\ h\neq\tr_{q/p}(w^{-1}v^2) \\
\end{array}\right.\\
&=&\left\{
\begin{array}{ll}
0 & \myatop{\mbox{ for $v=2^{-1}u,\ h=\tr_{q/p}(w^{-1}v^2),$}}{\mbox{ or $v\neq 2^{-1}u$},}\\
G(\eta,\chi)G(\eta',\chi')q & \myatop{\mbox{ for $v=2^{-1}u,\ h\neq\tr_{q/p}(w^{-1}v^2),$}}{\mbox{$\eta(w)\eta'(h-\tr_{q/p}(w^{-1}v^2))=1,$}}\\
-G(\eta,\chi)G(\eta',\chi')q & \myatop{\mbox{ for $v=2^{-1}u,\ h\neq\tr_{q/p}(w^{-1}v^2),$}}{\mbox{$\eta(w)\eta'(h-\tr_{q/p}(w^{-1}v^2))=-1.$}}\\
\end{array}\right.\\
\end{eqnarray*}
\end{enumerate}

When $w=0$ and $(u,v)\neq (0,0)$, it is easy to deduce that $\Omega=0$.

From the discussions above and Lemma \ref{quadGuasssum}, we have
\begin{eqnarray*}
N(u,v,w,h)&=&\left\{
\begin{array}{ll}
p^{2m-1} & \myatop{\mbox{for $v\neq 2^{-1}u,\ w\neq 0,$}}{\mbox{ or $w=0,\ (u,v)\neq (0,0)$,}}\\
p^{m-1}\left(p^m-(p-1)p^{\frac{m}{2}}(\sqrt{-1})^{\frac{p-1}{2}m}\right) & \myatop{\mbox{ for $v=2^{-1}u,\ h=\tr_{q/p}(w^{-1}v^2),$}}{\mbox{$\eta(w)=1,$}}\\
p^{m-1}\left(p^m+p^{\frac{m}{2}}(\sqrt{-1})^{\frac{p-1}{2}m}\right) & \myatop{\mbox{ for $v=2^{-1}u,\ h\neq\tr_{q/p}(w^{-1}v^2),$}}{\mbox{$\eta(w)=1,$}}\\
p^{m-1}\left(p^m+(p-1)p^{\frac{m}{2}}(\sqrt{-1})^{\frac{p-1}{2}m}\right) & \myatop{\mbox{ for $v=2^{-1}u,\ h=\tr_{q/p}(w^{-1}v^2),$}}{\mbox{$\eta(w)=-1,$}}\\
p^{m-1}\left(p^m-p^{\frac{m}{2}}(\sqrt{-1})^{\frac{p-1}{2}m}\right) & \myatop{\mbox{ for $v=2^{-1}u,\ h\neq\tr_{q/p}(w^{-1}v^2),$}}{\mbox{$\eta(w)=-1$}}\\
\end{array}\right.
\end{eqnarray*}
for even $m$, and
\begin{eqnarray*}
N(u,v,w,h)&=&\left\{
\begin{array}{ll}
p^{2m-1}, & \myatop{\mbox{ for $v=2^{-1}u,\ h=\tr_{q/p}(w^{-1}v^2),\ w\neq 0$ or}}{\mbox{$v\neq 2^{-1}u,\ w\neq 0$ or $w=0,\ (u,v)\neq (0,0),$}}\\
p^{m-1}\left(p^m+p^{\frac{m+1}{2}}(\sqrt{-1})^\frac{(p-1)(m+1)}{2}\right) & \myatop{\mbox{ for $v=2^{-1}u,\ h\neq\tr_{q/p}(w^{-1}v^2),\ w\neq 0,$}}{\mbox{$\eta(w)\eta'(h-\tr_{q/p}(w^{-1}v^2))=1,$}}\\
p^{m-1}\left(p^m-p^{\frac{m+1}{2}}(\sqrt{-1})^\frac{(p-1)(m+1)}{2}\right) & \myatop{\mbox{ for $v=2^{-1}u,\ h\neq\tr_{q/p}(w^{-1}v^2),\ w\neq 0,$}}{\mbox{$\eta(w)\eta'(h-\tr_{q/p}(w^{-1}v^2))=-1$}}\\
\end{array}\right.\\
\end{eqnarray*}
for odd $m$. The Hamming weight of any codeword $$\bc(u,v,w,h):=\left(\left(\tr_{p^m/p}(g(x,y))+h\right)_{(x,y)\in \gf(p^m)^{2}},\tr_{p^m/p}(w)\right)\in \C_{\V}^{(p)}$$ is then given by 
\begin{eqnarray*}
\lefteqn{ \wt(\bc(u,v,w,h)) } \\
&=&\left\{
\begin{array}{ll}
0 & \mbox{ for $(u,v,w,h)=(0,0,0,0)$} \\
p^{2m} & \mbox{ for $(u,v,w)=(0,0,0),\ h\neq 0$}\\
p^{2m-1}(p-1) & \myatop{\mbox{ for $v\neq 2^{-1}u,\ w\neq 0,\ \tr_{p^m/p}(w)=0$,}}{\mbox{ or $w=0,\ (u,v)\neq (0,0)$}}\\
p^{2m-1}(p-1)+1 & \mbox{ for $v\neq 2^{-1}u,\ w\neq 0,\ \tr_{p^m/p}(w)\neq0$} \\
p^{m-1}(p-1)\left(p^m+p^{\frac{m}{2}}(\sqrt{-1})^{\frac{p-1}{2}m}\right) & \myatop{\mbox{ for $v=2^{-1}u,\ h=\tr_{q/p}(w^{-1}v^2),$}}{\mbox{$\eta(w)=1,\ \tr_{p^m/p}(w)=0$}} \\
p^{m-1}(p-1)\left(p^m+p^{\frac{m}{2}}(\sqrt{-1})^{\frac{p-1}{2}m}\right)+1 & \myatop{\mbox{ for $v=2^{-1}u,\ h=\tr_{q/p}(w^{-1}v^2),$}}{\mbox{$\eta(w)=1,\ \tr_{p^m/p}(w)\neq0$}} \\
p^{m-1}\left(p^m(p-1)-p^{\frac{m}{2}}(\sqrt{-1})^{\frac{p-1}{2}m}\right) & \myatop{\mbox{ for $v=2^{-1}u,\ h\neq\tr_{q/p}(w^{-1}v^2),$}}{\mbox{$\eta(w)=1,\ \tr_{p^m/p}(w)=0$}} \\
p^{m-1}\left(p^m(p-1)-p^{\frac{m}{2}}(\sqrt{-1})^{\frac{p-1}{2}m}\right)+1 & \myatop{\mbox{ for $v=2^{-1}u,\ h\neq\tr_{q/p}(w^{-1}v^2),$}}{\mbox{$\eta(w)=1,\ \tr_{p^m/p}(w)\neq0$}} \\
p^{m-1}(p-1)\left(p^m-p^{\frac{m}{2}}(\sqrt{-1})^{\frac{p-1}{2}m}\right) & \myatop{\mbox{ for $v=2^{-1}u,\ h=\tr_{q/p}(w^{-1}v^2),$}}{\mbox{$\eta(w)=-1,\ \tr_{p^m/p}(w)=0$}} \\
p^{m-1}(p-1)\left(p^m-p^{\frac{m}{2}}(\sqrt{-1})^{\frac{p-1}{2}m}\right)+1 & \myatop{\mbox{ for $v=2^{-1}u,\ h=\tr_{q/p}(w^{-1}v^2),$}}{\mbox{$\eta(w)=-1,\ \tr_{p^m/p}(w)\neq0$}} \\
p^{m-1}\left(p^m(p-1)+p^{\frac{m}{2}}(\sqrt{-1})^{\frac{p-1}{2}m}\right) & \myatop{\mbox{ for $v=2^{-1}u,\ h\neq\tr_{q/p}(w^{-1}v^2),$}}{\mbox{$\eta(w)=-1,\ \tr_{p^m/p}(w)=0$}} \\
p^{m-1}\left(p^m(p-1)+p^{\frac{m}{2}}(\sqrt{-1})^{\frac{p-1}{2}m}\right)+1 & \myatop{\mbox{ for $v=2^{-1}u,\ h\neq\tr_{q/p}(w^{-1}v^2),$}}{\mbox{$\eta(w)=-1,\ \tr_{p^m/p}(w)\neq0$}} \\
\end{array}\right.\\
&=&\left\{
\begin{array}{ll}
0 & \mbox{ with 1 time,}\\
p^{2m} & \mbox{ with $p-1$ times,}\\
p^{2m-1}(p-1) & \mbox{ with $(p^m-1)(p^{2m}+p)$ times,}\\
p^{2m-1}(p-1)+1 & \mbox{ with $p^{2m}(p^m-1)(p-1)$ times,}\\
p^{m-1}(p-1)\left(p^m+p^{\frac{m}{2}}(\sqrt{-1})^{\frac{p-1}{2}m}\right) & \mbox{ with $\frac{p^{2m-1}-p^m-(p-1)p^{\frac{3m-2}{2}}(\sqrt{-1})^{\frac{(p-1)m}{2}}}{2}$ times,}\\
p^{m-1}(p-1)\left(p^m+p^{\frac{m}{2}}(\sqrt{-1})^{\frac{p-1}{2}m}\right)+1 & \mbox{ with $\frac{(p-1)(p^{2m-1}+p^{\frac{3m-2}{2}}(\sqrt{-1})^{\frac{(p-1)m}{2}})}{2}$ times,}\\
p^{m-1}\left(p^m(p-1)-p^{\frac{m}{2}}(\sqrt{-1})^{\frac{p-1}{2}m}\right) & \mbox{ with $\frac{(p-1)(p^{2m-1}-p^m-(p-1)p^{\frac{3m-2}{2}}(\sqrt{-1})^{\frac{(p-1)m}{2}})}{2}$ times,}\\
p^{m-1}\left(p^m(p-1)-p^{\frac{m}{2}}(\sqrt{-1})^{\frac{p-1}{2}m}\right)+1 & \mbox{ with $\frac{(p-1)^2(p^{2m-1}+p^{\frac{3m-2}{2}}(\sqrt{-1})^{\frac{(p-1)m}{2}})}{2}$ times,}\\
p^{m-1}(p-1)\left(p^m-p^{\frac{m}{2}}(\sqrt{-1})^{\frac{p-1}{2}m}\right) & \mbox{ with $\frac{p^{2m-1}-p^m+(p-1)p^{\frac{3m-2}{2}}(\sqrt{-1})^{\frac{(p-1)m}{2}}}{2}$ times,}\\
p^{m-1}(p-1)\left(p^m-p^{\frac{m}{2}}(\sqrt{-1})^{\frac{p-1}{2}m}\right)+1 & \mbox{ with $\frac{(p-1)(p^{2m-1}-p^{\frac{3m-2}{2}}(\sqrt{-1})^{\frac{(p-1)m}{2}})}{2}$ times,}\\
p^{m-1}\left(p^m(p-1)+p^{\frac{m}{2}}(\sqrt{-1})^{\frac{p-1}{2}m}\right) & \mbox{ with $\frac{(p-1)(p^{2m-1}-p^m+(p-1)p^{\frac{3m-2}{2}}(\sqrt{-1})^{\frac{(p-1)m}{2}})}{2}$ times,}\\
p^{m-1}\left(p^m(p-1)+p^{\frac{m}{2}}(\sqrt{-1})^{\frac{p-1}{2}m}\right)+1 & \mbox{ with $\frac{(p-1)^2(p^{2m-1}-p^{\frac{3m-2}{2}}(\sqrt{-1})^{\frac{(p-1)m}{2}})}{2}$ times}
\end{array}\right.\\
\end{eqnarray*}
for even $m$, where the frequency of each weight can be easily derived with Lemma \ref{lem-equalities}, and the Hamming weight 
\begin{eqnarray*}
\lefteqn{ \wt(\bc(u,v,w,h)) } \\
&=&\left\{
\begin{array}{ll}
0 & \mbox{ for $(u,v,w,h)=(0,0,0,0)$}\\
p^{2m} & \mbox{ for $(u,v,w)=(0,0,0),\ h\neq 0$} \\
p^{2m-1}(p-1) & \myatop{\myatop{\mbox{ for $v=2^{-1}u,\ h=\tr_{q/p}(w^{-1}v^2),$}}{\mbox{$w\neq 0,\ \tr_{p^m/p}(w)=0$ or}}}{\myatop{\mbox{$v\neq 2^{-1}u,\ w\neq 0,\ \tr_{p^m/p}(w)=0$}}{\mbox{or $w=0,\ (u,v)\neq (0,0)$}}} \\
p^{2m-1}(p-1)+1 & \myatop{\myatop{\mbox{ for $v=2^{-1}u,\ h=\tr_{q/p}(w^{-1}v^2),$}}{\mbox{$\tr_{p^m/p}(w)\neq0$ or}}}{\mbox{$v\neq 2^{-1}u,\ \tr_{p^m/p}(w)\neq0$}} \\
p^{\frac{3m-1}{2}}\left(p^{\frac{m-1}{2}}(p-1)-(\sqrt{-1})^\frac{(p-1)(m+1)}{2}\right) & \myatop{\mbox{ for $v=2^{-1}u,\ h\neq\tr_{q/p}(w^{-1}v^2),\ w\neq 0,$}}{\myatop{\mbox{$\eta(w)\eta'(h-\tr_{q/p}(w^{-1}v^2))=1,$}}{\mbox{$\tr_{p^m/p}(w)=0$}}} \\
p^{\frac{3m-1}{2}}\left(p^{\frac{m-1}{2}}(p-1)-(\sqrt{-1})^\frac{(p-1)(m+1)}{2}\right)+1 & \myatop{\mbox{ for $v=2^{-1}u,\ h\neq\tr_{q/p}(w^{-1}v^2),$}}{\myatop{\mbox{$\eta(w)\eta'(h-\tr_{q/p}(w^{-1}v^2))=1,$}}{\mbox{$\tr_{p^m/p}(w)\neq0$}}} \\
p^{\frac{3m-1}{2}}\left(p^{\frac{m-1}{2}}(p-1)+(\sqrt{-1})^\frac{(p-1)(m+1)}{2}\right) & \myatop{\mbox{ for $v=2^{-1}u,\ h\neq\tr_{q/p}(w^{-1}v^2),\ w\neq 0,$}}{\myatop{\mbox{$\eta(w)\eta'(h-\tr_{q/p}(w^{-1}v^2))=-1,$}}{\mbox{$\tr_{p^m/p}(w)=0$}}} \\
p^{\frac{3m-1}{2}}\left(p^{\frac{m-1}{2}}(p-1)+(\sqrt{-1})^\frac{(p-1)(m+1)}{2}\right)+1 & \myatop{\mbox{ for $v=2^{-1}u,\ h\neq\tr_{q/p}(w^{-1}v^2),$}}{\myatop{\mbox{$\eta(w)\eta'(h-\tr_{q/p}(w^{-1}v^2))=-1,$}}{\mbox{$\tr_{p^m/p}(w)\neq 0$}}} \\
\end{array}\right.\\
&=&\left\{
\begin{array}{ll}
0 & \mbox{ with 1 time,}\\
p^{2m} & \mbox{ with $p-1$ times,}\\
p^{2m-1}(p-1) & \myatop{\mbox{ with $p^m(p^{m-1}-1)(p^{m+1}-p+1)$}}{\mbox{$+p(p^{2m}-1)$ times,}}\\
p^{2m-1}(p-1)+1 & \mbox{ with $p^{2m-1}(p-1)(p^{m+1}-p+1)$ times,}\\
p^{\frac{3m-1}{2}}\left(p^{\frac{m-1}{2}}(p-1)-(\sqrt{-1})^\frac{(p-1)(m+1)}{2}\right) & \mbox{ with $\frac{p^m(p^{m-1}-1)(p-1)}{2}$ times,}\\
p^{\frac{3m-1}{2}}\left(p^{\frac{m-1}{2}}(p-1)-(\sqrt{-1})^\frac{(p-1)(m+1)}{2}\right)+1 & \mbox{ with $\frac{p^{2m-1}(p-1)^2}{2}$ times,}\\
p^{\frac{3m-1}{2}}\left(p^{\frac{m-1}{2}}(p-1)+(\sqrt{-1})^\frac{(p-1)(m+1)}{2}\right) & \mbox{ with $\frac{p^m(p^{m-1}-1)(p-1)}{2}$ times,}\\
p^{\frac{3m-1}{2}}\left(p^{\frac{m-1}{2}}(p-1)+(\sqrt{-1})^\frac{(p-1)(m+1)}{2}\right)+1 & \mbox{ with $\frac{p^{2m-1}(p-1)^2}{2}$ times}\\
\end{array}\right.\\
\end{eqnarray*}
for odd $m$, where the frequency of each weight can be easily determined.

Note that the dimension is $3m+1$ as $A_0=1$ whether $m$ is even or odd. Then the desired conclusions follow.
\end{proof} 

\begin{example} 
Let $\cV$ be the elliptic quadric. 
\begin{enumerate}
\item Let $m=2$ and $w$ be a generator of $\gf(2^3)$ with $w^2+w+1=0$, and $a=w^3$. Then 
      $\C_{\V}^{(2)}$ has parameters $[17, 7, 6]$ and its dual has parameters $[17, 10, 4]$. 
\item Let $m=3$ and $w$ be a generator of $\gf(2^3)$ with $w^3+w+1=0$, and $a=w^3$. Then 
      $\C_{\V}^{(2)}$ has parameters $[65, 10, 28]$ and its dual has parameters $[65, 55, 4]$. 
\item Let $m=2$ and $w$ be a generator of $\gf(3^2)$ with $w^2+2w+2=0$, and $a=w^3$. Then 
      $\C_{\V}^{(3)}$ has parameters $[82, 7, 51]$ and its dual has parameters $[82, 75, 4]$. 
\end{enumerate} 
All of these codes and their duals are optimal according to the tables of best codes known 
maintained at http://www.codetables.de. 
\end{example}

At the end of this section, we explain why the subfield codes of ovoid codes are interesting. 
It is known that the set $\V$ of (\ref{eqn-ellipticquadric}) is an ovoid if and only 
if $x^2+x+a$ is irreducible over $\gf(q)$. The parameters of the subfield code $\C_{\V}^{(p)}$ 
of the code $\C_{\V}$ were determined for all $a$. In all cases, the code $\C_{\V}^{(p)}$ has 
length $p^{2m}+1$ and dimension $3m+1$. However, its minimum distance $d^{(p)}$ and weight distribution vary according to $a$ for odd $p$. Specifically, we have the following for odd 
$p$. 
\begin{itemize}
\item If $x^2+x+a$ is irreducible, then $\V$ is an ovoid and 
      $$ 
        d^{(p)}= p^{2m-1}(p-1)-p^{m-1}.  
      $$ 
      Further, the dual code $\C_{\V}^{(p)\perp}$ has minimum distance $d^{(p)\perp}=4$. 
\item If $x^2+x+a$ is reducible and $a \neq 1/4$, then $\V$ is not an ovoid and 
      $$ 
        d^{(p)}= p^{2m-1}(p-1)-p^{m-1}(p-1).  
      $$ 
      Further, the dual code $\C_{\V}^{(p)\perp}$ has minimum distance $3$ according to 
      our experimental data.  
\item If $m$ is even and $a = 1/4$, then $x^2+x+a$ is reducible, $\V$ is not an ovoid and 
      $$ 
        d^{(p)}= p^{2m-1}(p-1)-p^{m-1}(p-1)p^{m/2}.  
      $$ 
      Further, the dual code $\C_{\V}^{(p)\perp}$ has minimum distance $3$ according to 
      our experimental data.             
\item If $m$ is odd and $a = 1/4$, then $x^2+x+a$ is reducible, $\V$ is not an ovoid and 
      $$ 
        d^{(p)}= p^{2m-1}(p-1)-p^{m-1}p^{(m+1)/2}.  
      $$ 
      Further, the dual code $\C_{\V}^{(p)\perp}$ has minimum distance $3$ according to 
      our experimental data.                   
\end{itemize}
Therefore, both $\C_{\V}^{(p)}$ and $\C_{\V}^{(p)\perp}$ have the best minimum distance 
when $\V$ is an ovoid. The comparison above shows that the subfield codes of ovoid codes 
are indeed interesting.

\section{Subfield codes of the Tits ovoid codes}

Let $q=2^{2e+1}$ with $e\geq 1$. Recall that the Tits ovoids are defined by 
 $$\T=\{(0,0,1,0)\}\cup \{(x,y,x^\sigma+xy+y^{\sigma+2},1):x,y\in \Bbb \gf(q)\},$$ where $\sigma=2^{e+1}$.
Denote
$$t_1(x,y)=x,\ t_2(x,y)=y,\ t_3(x,y)=x^\sigma+xy+y^{\sigma+2}$$ and
$$G_{x,y}=\begin{bmatrix} t_1(x,y)\\ t_2(x,y)\\ t_3(x,y) \\ 1 \end{bmatrix}_{(x,y)\in \gf(q)^2} $$
which is a $4\times q^2$ matrix over $\gf(q)$. The Tits ovoid code $\C_{\T}$ over $\gf(q)$ 
has the generator matrix
$$G_{\T}=\begin{bmatrix} G_{x,y}  \begin{array}{c}
0 \\
0\\
1\\
0\end{array}
\end{bmatrix}.$$
Our task in this section is to investigate the subfield code $\C_{\T}^{(2)}$ of the 
Tits ovoid code $\C_{\T}$.

Using the definition of $G_{\T}$ and Theorem \ref{th-tracerepresentation}, we have the
following trace representation of $\C_{\T}^{(2)}$: 
\begin{eqnarray*}
\C_{\T}^{(2)}&=&\left\{\left(\left(\tr_{q/2}(ut_1(x,y)+vt_2(x,y)+wt_3(x,y))+h\right)_{(x,y)\in \gf(q)^{2}},\tr_{q/2}(w)\right):\myatop {u,v,w \in \gf(q)}{h\in \gf(2)}\right\}\\
&=&\left\{\left(\left(\tr_{q/2}(t(x,y))+h\right)_{(x,y)\in \gf(q)^{2}},\tr_{q/2}(w)\right):\myatop {u,v,w \in \gf(q)}{h\in \gf(2)}\right\}
\end{eqnarray*}
where $t(x,y):=ut_1(x,y)+vt_2(x,y)+wt_3(x,y))=ux+vy+wx^\sigma+wxy+wy^{\sigma+2}$.

\begin{theorem}\label{th-anotherovoidcode}
Let $q=2^{2e+1}$ with $e\geq 1$. Then $\C_{\T}^{(2)}$ is a linear code with parameters $[2^{4e+2}+1,6e+4]$ and the weight distribution in Table \ref{tab-6}. Its dual $\C_{\T}^{(2)\perp}$ has parameters $[2^{4e+2}+1,2^{4e+2}-6e-3,4]$.
\end{theorem}

\begin{table}[ht]
\begin{center}
\caption{The weight distribution of $\C_{\T}^{(2)}$}\label{tab-6}
\begin{tabular}{cc} \hline
Weight  &  Multiplicity   \\ \hline
0    &   1\\
$2^{4e+2}$    &   1\\
$2^{4e+1}$ & $2(2^{4e+2}-1)$\\
$2^{4e+1}+2^{2e}$ & $2^{4e+2}(2^{2e}-1)$\\
$2^{4e+1}-2^{2e}$ & $2^{4e+2}(2^{2e}-1)$\\
$2^{4e+1}+2^{2e}+1$ & $2^{6e+2}$\\
$2^{4e+1}-2^{2e}+1$ & $2^{6e+2}$\\
\hline
\end{tabular}
\end{center}
\end{table}

\begin{proof}
Let $\chi$ be the canonical additive character of $\gf(q)$.
Firstly, assume that $(u,v,w)\neq (0,0,0)$.
Denote
$$N_0(u,v,w)=\sharp\{(x,y)\in \gf(q)^2:\tr_{q/2}(t(x,y))=0\}$$ and $$N_1(u,v,w)=\sharp\{(x,y)\in \gf(q)^2:\tr_{q/2}(t(x,y))=1\}.$$ By the orthogonality relation of additive characters, we have
\begin{eqnarray}\label{eqn-11}
\nonumber 2N_0(u,v,w)&=&\sum_{(x,y)\in \gf(q)^2}\sum_{z\in \gf(2)}(-1)^{z\tr_{q/2}(t(x,y))}\\
&=&q^2+\sum_{(x,y)\in \gf(q)^2}(-1)^{\tr_{q/2}(t(x,y))}.
\end{eqnarray}
We discuss the value of $N_0(u,v,w)$ in the following cases.
\begin{enumerate}
\item If $w=0$, we have $t(x,y)=ux+vy$. Since $(u,v)\neq (0,0)$, we deduce that $N_0(u,v,w)=q^2/2=2^{4e+1}$.
\item If $w\neq 0$, we denote
$$\Delta=\sum_{(x,y)\in \gf(q)^2}(-1)^{\tr_{q/2}(t(x,y))}.$$ Then
\begin{eqnarray*}
\Delta^2&=&\left(\sum_{(x,y)\in \gf(q)^2}(-1)^{\tr_{q/2}(-t(x,y))}\right)\left(\sum_{(x_1,y_1)\in \gf(q)^2}(-1)^{\tr_{q/2}(t(x_1,y_1))}\right)\\
&=&\sum_{(x,y)\in \gf(q)^2}\sum_{(x_1,y_1)\in \gf(q)^2}(-1)^{\tr_{q/2}(t(x_1,y_1)-t(x,y))}\\
&=&\sum_{(x,y)\in \gf(q)^2}\sum_{(A,B)\in \gf(q)^2}(-1)^{\tr_{q/2}(t(x+A,y+B)-t(x,y))}\\
&=&q^2-\sum_{(A,B)\in \gf(q)^2\setminus \{(0,0)\}}\sum_{(x,y)\in \gf(q)^2}(-1)^{\tr_{q/2}(t(x+A,y+B)-t(x,y))}\\
&=&q^2-\sum_{(A,B)\in \gf(q)^2\setminus \{(0,0)\}}\chi(uA+vB+wAB+wA^{2^{e+1}}+wB^{2^{e+1}+2})\\
& &\times \sum_{y\in \gf(q)}\chi(wB^2y^{2^{e+1}}+wB^{2^{e+1}}y^2+wAy)\sum_{x\in \gf(q)}\chi(wBx)\\
&=&q^2-q\sum_{A\in \gf(q)^*}\chi(uA+wA^{2^{e+1}})\sum_{y\in \gf(q)}\chi(wAy)\\
&=&q^2,
\end{eqnarray*}
where we used the variable substitution $x_1=x+A,y_1=y+B$ in the third equality and the last two equalities hold due to the orthogonality relation of additive characters. Hence $\Delta=\pm q$ and Equation (\ref{eqn-11}) implies
$$N_0(u,v,w)=2^{4e+1}\pm 2^{2e}.$$
\end{enumerate}
Combining the two cases above yields
\begin{eqnarray*}
N_0(u,v,w)=\left\{
\begin{array}{ll}
2^{4e+1}    &   \mbox{ for }w=0,\\
2^{4e+1}\pm 2^{2e}    &   \mbox{ for }w\neq 0,\\
\end{array} \right.
\end{eqnarray*}
where $(u,v,w)\neq (0,0,0)$ and $N_1(u,v,w)=2^{4e+2}-N_0(u,v,w)$.

For any codeword $\bc(u,v,w,h):=\left(\left(\tr_{q/2}(t(x,y))+h\right)_{(x,y)\in \gf(q)^{2}},\tr_{2^m/2}(w)\right)\in \C_{\T}^{(2)}$, by the discussions above we deduce that
\begin{eqnarray*}
\wt(\bc(u,v,w,h))&=&\left\{
\begin{array}{ll}
0    &   \mbox{ for }(u,v,w,h)=(0,0,0,0) \\
2^{4e+2}    &   \mbox{ for }(u,v,w,h)=(0,0,0,1) \\
N_1(u,v,w) & \mbox{ for }w=h=0,\ (u,v)\neq (0,0) \\
N_0(u,v,w) & \mbox{ for }w=0,\ h=1,\ (u,v)\neq (0,0) \\
N_1(u,v,w) & \mbox{ for }h=0,\ w\neq 0, \tr_{q/2}(w)=0,\ (u,v)\in\gf(q)^2 \\
N_0(u,v,w) & \mbox{ for }h=1,\ w\neq 0, \tr_{q/2}(w)=0,\ (u,v)\in \gf(q)^2 \\
N_1(u,v,w)+1 & \mbox{ for }h=0,\ \tr_{q/2}(w)\neq0,\ (u,v)\in \gf(q)^2 \\
N_0(u,v,w)+1 & \mbox{ for }h=1,\ \tr_{q/2}(w)\neq0,\ (u,v)\in \gf(q)^2 \\
\end{array} \right.\\
&=&\left\{\begin{array}{ll}
0    &   \mbox{ with 1 time},\\
2^{4e+2}    &   \mbox{ with 1 time},\\
2^{4e+1} & \mbox{ with }2(2^{4e+2}-1)\ \mbox{times},\\
2^{4e+1}+2^{2e} & \mbox{ with }2^{4e+2}(2^{2e}-1)\ \mbox{times},\\
2^{4e+1}-2^{2e} & \mbox{ with }2^{4e+2}(2^{2e}-1)\ \mbox{times},\\
2^{4e+1}+2^{2e}+1 & \mbox{ with }2^{6e+2}\ \mbox{times},\\
2^{4e+1}-2^{2e}+1 & \mbox{ with }2^{6e+2}\ \mbox{times},
\end{array} \right.
\end{eqnarray*}
where the frequency of each weight is easy to derive. The dimension is $3m+1$ as $A_0=1$.

The parameters of its dual follow from Theorem \ref{th-dualdistanceofC} and the sphere-packing bound. 
\end{proof}

\section{Concluding remarks} 

Example 1 demonstrates that the subfield code $\C_{\cO}^{(p)}$ of some ovoid code 
$\C_{\cO}$ is optimal. When $\cO$ is an elliptic quadric or the Tits ovoid, the dual 
code $\C_{\cO}^{(p)\perp}$ is distance-optimal according to the sphere-packing bound.  

Let $q=2^m$. Note that every ovoid code $\C$ over $\gf(q)$ must have parameters 
$[q^2+1, 4, q^2-q]$ and the 
weight enumerator 
\begin{eqnarray*} 
1+(q^2-q)(q^2+1)z^{q^2-q}+(q-1)(q^2+1)z^{q^2}. 
\end{eqnarray*} 
However, the subfield code $\C^{(2)}$ may have different parameters and weight distributions. 
In the case of the elliptic quadric and Tits ovoid, the subfield code $\C^{(2)}$ has the same 
parameters and weight distribution (only six nonzero weights). However, the subfield code 
$\C^{(2)}$ of another family of ovoid codes documented in \cite{Ding18arxiv} has $2^m$ 
nonzero weights and very different parameters according to our Magma experimental data. 
It seems very difficult to settle the parameters of the subfield codes of the ovoid codes 
presented in \cite{Ding18arxiv}. 

Finally, we point out that $m$-ovoids are related to ovoids and give two-weight codes \cite{FMX}. It would be interesting to study the subfield codes of the two-weight codes from $m$-ovoids.

\end{document}